\documentclass{vldb}

\usepackage{appendix}
\usepackage{amsmath}
\usepackage{amsfonts}
\usepackage{amssymb}
\usepackage{srcltx}
\usepackage{multirow}
\usepackage{boxedminipage}

\usepackage{xspace}
\usepackage{graphicx}
\usepackage{color}
\definecolor{DarkGreen}{rgb}{0.1,0.5,0.1}
\definecolor{DarkRed}{rgb}{0.5,0.1,0.1}
\definecolor{DarkBlue}{rgb}{0.1,0.1,0.5}
\definecolor{White}{rgb}{1.0,1.0,1.0}

\newcommand{\avg}{\mathop{\textrm{avg}}}
\newcommand{\smlsize}[1]{\fontsize{9}{9}\selectfont #1}

\def\draft{1} 

\def\submit{0} 

\ifnum\draft=1 
    \def\ShowAuthNotes{1}
\else
    \def\ShowAuthNotes{0}
\fi

\ifnum\submit=1
\newcommand{\forsubmit}[1]{#1}
\newcommand{\forreals}[1]{}
\else
\newcommand{\forreals}[1]{#1}
\newcommand{\forsubmit}[1]{}
\fi

\ifnum\ShowAuthNotes=1
\newcommand{\authnote}[2]{{ \footnotesize \bf{\color{red}[#1's Note: {\color{blue}#2}]}}}
\else
\newcommand{\authnote}[2]{}
\fi

\newtheorem{theorem}{Theorem}[section]

\newtheorem{lemma}{Lemma}[section]
\newtheorem{corollary}[lemma]{Corollary}

\newtheorem{fact}[lemma]{Fact}

\newtheorem{definition}{Definition}[section]

\newcommand{\Psymb}{\mathbb{P}}

\DeclareMathOperator*{\ProbOp}{\Psymb r}
\renewcommand{\Pr}{\ProbOp}

\newcommand{\mper}{\,.}
\newcommand{\mcom}{\,,}

\newcommand{\remove}[1]{}
\newcommand{\eps}{\varepsilon}
\renewcommand{\epsilon}{\varepsilon}
\newcommand{\Laplace}{\mathit{Lap}}

\newcommand{\defeq}{\stackrel{\small \mathrm{def}}{=}}

\newcommand{\err}{\mathrm{maxerr}}
\newcommand{\adderr}{\mathrm{adderr}}

\title{
A Simple and Practical Algorithm \\ for
Differentially Private Data Release
}

\numberofauthors{3}

\author{
\alignauthor
Moritz Hardt\\
\email{mhardt@us.ibm.com}
\and
\alignauthor
Katrina Ligett\\
\email{katrina@caltech.edu}
\and
\alignauthor
Frank McSherry\\
\email{mcsherry@microsoft.com}
}

\begin{document}
\maketitle

\begin{abstract}
  We present a new algorithm for differentially private data release,
  based on a simple combination of the Exponential Mechanism with the
  Multiplicative Weights update rule.  Our \emph{MWEM} algorithm achieves what are
  the best known and nearly optimal theoretical guarantees, while at
  the same time being simple to implement and experimentally more
  accurate on actual data sets than existing techniques.
%
\end{abstract}


\section{Introduction}

Sensitive statistical data on individuals are ubiquitous, and publishable
analysis of such private data is an important objective.
When releasing statistics or synthetic data based on sensitive
data sets, one must balance the inherent tradeoff between the usefulness of
the released information and the privacy of the affected individuals.  Against
this backdrop, differential privacy \cite{DN03,DN04,DMNS06} has emerged as a
compelling privacy definition that allows one to understand this tradeoff via
formal, provable guarantees. In recent years, the theoretical
literature on differential privacy has provided a large repertoire of
techniques for achieving the definition in a variety of settings (see,
e.g.,~\cite{Dwork09,Dwork11}).
However, data analysts have found that several
algorithms for achieving differential privacy add unacceptable levels of noise.




In this work we develop a broadly applicable, simple, and
easy-to-implement algorithm, capable of substantially improving performance on
many realistic datasets. The algorithm is a combination of the Multiplicative
Weights approach of \cite{HR10,GuptaHRU11}, maintaining and correcting an approximating
dataset through queries on which the approximate and true datasets differ, and
the Exponential Mechanism \cite{MT07}, which selects the queries most
informative to the Multiplicative Weights algorithm (specifically, those most
incorrect vis-a-vis the current approximation). While in the worst case one
must separately measure all required queries, for less adversarial data
and query sets our approximation can provide very accurate answers to all queries, having privately measured only a relatively small subset of them.

\subsection{Our Results}

We present \emph{MWEM}, a new differentially private algorithm producing
synthetic datasets (formally: a fractional weighting of the domain)
respecting any set of \emph{linear queries} (those that apply a function
to each record and sum the results). Our algorithm matches the best
known (and nearly optimal) theoretical accuracy guarantees for releasing
differentially private answers to a set of counting queries (those
mapping records to $\{0,1\})$. 

We present experimental results for producing differentially private
synthetic data for a variety of problems studied in prior work, based
on a variety of real-world data sets. In each case, we empirically evaluate the
accuracy of the differentially private data produced by MWEM using the same query class and accuracy metric proposed by the corresponding prior work.
\begin{enumerate}
\item We consider range queries as studied by~\cite{LiMi11,LiMiAdaptive12},
and find up to three orders of magnitude improvements in accuracy for fixed privacy parameters across several
datasets. 
\item We investigate contingency table release across a collection
of statistical benchmarks as in~\cite{FRY10} and find we are able to improve on prior
work for each. 
\item We consider datacube release 
as studied by \cite{DingWHL11} and find that our general-purpose
algorithm improves on specialized algorithms designed to optimize various criteria.
\end{enumerate}
Releasing synthetic data guarantees important properties, including consistency of statistics and compatibility with downstream analyses expecting actual datasets as input. 

Finally, we describe a scalable implementation of MWEM capable
of processing and releasing datasets of substantial complexity. 
Producing synthetic data for the classes of queries we consider is known to be computationally hard in the \emph{worst-case}~\cite{DworkNRRV09,UllmanV11}.
Indeed, almost all prior work performs computation proportional to the size of the
data domain, which limits them to datasets with relatively few attributes. 
In contrast, we are able to process datasets with
thousands of attributes, corresponding to domains of size $2^{1000}$. Our implementation integrates a scalable parallel
implementation of Multiplicative Weights, and a representation of the
approximating dataset in a factored form that only exhibits complexity when the
model requires it.

\pagebreak
\section{Our Approach}

MWEM maintains an approximating dataset as a (scaled) distribution
over the domain $D$ of data records. We repeatedly improve the
accuracy of this approximation with respect to the private dataset and
the desired query set by selecting and posing a query poorly served by
our approximation and improving the approximation to better reflect
the answer to this query. We select and pose queries using the
Exponential~\cite{MT07} and Laplace Mechanisms~\cite{DMNS06},
whose definitions and privacy properties we review in
Subsection~\ref{subsec:dp}. We improve our approximation using the
Multiplicative Weights update rule~\cite{HR10}, reviewed in
Subsection~\ref{subsec:mw}. We describe their integration in
Subsection~\ref{subsec:together}, giving a full description of our
algorithm and its formal properties.

\subsection{Differential Privacy and Mechanisms}\label{subsec:dp}

Differential privacy is a constraint on a randomized computation that the
computation should not reveal specifics of individual records present in the
input. It places this constraint by requiring the mechanism to behave almost
identically on any two datasets that are sufficiently close.

Imagine a dataset $A$ whose records are drawn from some abstract domain
$D$, and which is described as a function from $D$ to the natural
numbers $\mathbb{N}$, with $A(x)$ indicating the frequency (number of
occurrences) of $x$ in the dataset. We use $\|A - B\|$ to indicate the
sum of the absolute values of difference in frequencies (how many
records would have to be added and removed to change $A$ into another
dataset $B$).

\begin{definition}[Differential Privacy]
A mechanism $M$ mapping datasets to distributions over an output space $R$ 
provides $(\epsilon,
\delta)$-differential privacy if for every
$S\subseteq R$ and
for all data sets $A, B$ where $\|A - B\| \le 1$, 
\begin{equation*}
\Pr[M(A) \in S] \leq e^{\epsilon} \Pr[M(B) \in S] + \delta \; .
\end{equation*}
If $\delta=0$ we say that $M$ provides
$\epsilon$-differential privacy.
\end{definition}

The Exponential Mechanism~\cite{MT07} is an $\epsilon$-differentially
private mechanism that can be used to select 
among the best of a discrete set of alternatives, where ``best'' is
defined by a function relating each alternative to the underlying
secret data. Formally, for a set of alternative results $R$, we
require a quality scoring function $s : dataset \times R \rightarrow
\mathbb{R}$, where $s(B, r)$ is interpreted as the quality of the
result $r$ for the dataset $B$. To guarantee $\epsilon$-differential
privacy, the quality function is required to satisfy a 
stability property: that for each result r the difference $|s(A,r) -
s(B,r)|$ is at most $\|A-B\|$. With this quality function in hand, the
Exponential Mechanism $E$ simply selects a result $r$ from the
distribution satisfying
$$ \Pr[E(B) = r] \propto \exp(\epsilon \times  s(B,r) / 2). $$
Intuitively, the mechanism selects result $r$ biased
exponentially by its quality score. The Exponential Mechanism takes time linear
in the number of possible results, as it needs to evaluate $s(B,r)$ once for each $r$.

The Laplace Mechanism is an $\epsilon$-differentially private mechanism which reports
approximate sums of bounded functions across a dataset. If $f$ is a
function from records to the interval $[-1,+1]$, the Laplace Mechanism
$L$ obeys
	$$\Pr[L(B) = r] \propto \exp\big{(}-\epsilon \times |r - \sum_{x
  \in D} f(x) \times B(x)|\big{)} \; . $$ 
        Although the Laplace Mechanism is an instance of the
        Exponential Mechanism, it can be implemented much more
        efficiently, by adding Laplace noise with parameter
        $1/\epsilon$ to the sum $\sum_{x \in D} f(x) \times B(x)$. As
        the Laplace distribution is exponentially concentrated, the
        Laplace Mechanism provides an excellent approximation to the true sum.

\subsection{Multiplicative Weights Update Rule}\label{subsec:mw}

The Multiplicative Weights approach has seen application in many areas
of computer science. Here we will use it as proposed in Hardt and
Rothblum~\cite{HR10}, to repeatedly improve an approximate distribution to
better reflect some true distribution. The intuition behind
Multiplicative Weights is that should we find a query whose answer on
the true data is much larger than its answer or the approximate data, we
should scale up the approximating Weights on records contributing positively and
scale down the Weights on records contributing negatively. If the true
answer is much less than the approximate answer, we should do the opposite.

More formally, let $q$ be a function mapping records to the interval $[-1,+1]$, and extended to a function of datasets by accumulating the sum of $q$ applied to the individual records. If $A$ and $B$ are distributions over the domain $D$ of records,  where $A$ is a synthetic distribution intended to approximate a true distribution $B$ with respect to query $q$, then the Multiplicative Weights update rule recommends updating the weight $A$ places on each record $x$ by:
	$$ A_{new}(x) \propto A(x) \times \exp(q(x) \times (q(B) - q(A)) / 2) \; . $$
The proportionality sign indicates that 
the approximation
should be renormalized after scaling. Hardt and Rothblum show that each time this rule
is applied, the relative entropy between
$A$ and $B$
decreases by an additive $(q(A) - q(B))^2$. As long as we can continue
to find queries on which the two disagree, we can continue to improve
the approximation. 

Although our algorithm will manipulate datasets, we can divide their frequencies by the numbers of records $n$ whenever we need to apply Multiplicative Weights updates, and renormalize to a fixed number of records (rather than to one, as we would for a distribution).

\subsection{Putting Things Together}\label{subsec:together}

Our combined approach is simple: we repeatedly use the Exponential
Mechanism to find queries $q$ that are poorly served by our current
approximation of the underlying private data, we measure these queries
using the Laplace Mechanism, and we then use this measurement to
improve our distribution using the Multiplicative Weights update rule.

We assume there exists a set
$Q$ of queries $q$, each functions from the record domain $D$ to the
interval $[-1,+1]$ and extended to datasets $A$ by $q(A) = \sum_{x \in
  D} q(x) \times A(x)$. The MWEM algorithm appears in Figure
\ref{fig:algorithm}.
\begin{figure}[tbp]
\begin{minipage}{7.75cm}
\hrule
\phantom{.}
\noindent{\bf Inputs:} Data set~$B$ over a universe~$D,$\\
\noindent\phantom{\bf Inputs:} Number of iterations $T\in\mathbb{N},$\\
\noindent\phantom{\bf Inputs:} Privacy parameter $\epsilon>0.$\\

\noindent Let $n$ denote $\|B\|$, the number of records in $B$.

\noindent Let $A_0$ denote $n$ times the uniform distribution over~$D$.

\noindent For iteration $i=1,...,T$:
\begin{enumerate}
\item {\it Exponential Mechanism:}
Sample a query $q_i\in Q$ using the Exponential Mechanism parametrized with epsilon value $\epsilon / 2T$ and the score function 
$$s_i(B, q) = |q(A_{i-1}) - q(B)| \; .$$
%

\item {\it Laplace Mechanism:} Let measurement $m_i = q_i(B) + \Laplace(2T/\epsilon).$

\item {\it Multiplicative Weights:}
Let $A_i$ be $n$ times the distribution whose entries satisfy
$$ A_i(x) \propto A_{i-1}(x) \times \exp(q_i(x) \times (m_i - q_i(A_{i-1})) / 2n) \; . $$
\end{enumerate}
\noindent {\bf Output:} $A = \avg_{i < T} A_i.$
\end{minipage}
\phantom{.}
\hrule
\phantom{.}
\caption{The MWEM algorithm.\label{fig:algorithm}}
\end{figure}
To clarify how MWEM works, and to highlight its simplicity, we present a full implementation in the Appendix, in Figure \ref{fig:code}.

\subsubsection{Formal Guarantees}
As indicated in the
introduction, the formal guarantees of MWEM represent the best known theoretical results on
differentially private synthetic data release.

We first describe the privacy properties of our algorithm.
\begin{theorem}\label{thm:privacy} MWEM satisfies $\epsilon$-differential privacy.
\end{theorem}
\begin{proof}
  The composition rules for differential privacy state that $\epsilon$
  values accumulate additively. We make $T$ calls to the Exponential
  Mechanism with parameter $(\epsilon/2T)$ and $T$ calls to the
  Laplace Mechanism with parameter $(\epsilon/2T)$, resulting in
  $\epsilon$-differential privacy.
\end{proof}

We now bound the worst-case performance of the algorithm, in terms of
the maximum error between $A$ and $B$ across all $q \in Q$. The
natural range for $q(A)$ is $[-n,+n]$, and we see that by increasing
$T$ beyond $4\log|D|$ we can bring the error asymptotically smaller
than $n$.
\begin{theorem}\label{thm:utility} 
\label{UTILITY}
For any dataset $B$, set of linear queries $Q$, $T \in \mathbb{N}$, and $\epsilon > 0$, with probability at least $1-2T/|Q|$, MWEM produces $A$ such that
\begin{equation*}\label{eq:utility}
\max_{q \in Q} |q(A) - q(B)|  \le  2n\sqrt{\frac{\log|D|}T}
+\frac{10T\log|Q|}{\epsilon } \; .
\end{equation*}
\end{theorem}
\begin{proof}
  The proof of this theorem is an integration of pre-existing analyses
  of both the Exponential Mechanism and the Multiplicative Weights
  update rule. We include it for completeness in the Appendix.
\end{proof}

By optimizing the choice of $T$ to minimize the above bound we have
the following corollary, stated asymptotically only because the
constants are messy rather than large.
\begin{corollary}\label{thm:main1}
For any $B$, $Q$, and $\epsilon > 0$, there exists a $T$ such that with probability at least $1-1/poly(|Q|)$, the $A$ returned by MWEM obeys
\[
\max_{q \in Q} |q(A) - q(B)| \textrm{ is }  O\left(n^{2/3}\left(\frac{\log|D|\log|Q|}{\epsilon}\right)^{1/3}\right)\mper
\]
\end{corollary}
\noindent

If we permit $(\epsilon,\delta)$-differential privacy with $\delta > 0$, we can use {\em k-fold adaptive composition}~\cite{DRV10} to conclude that MWEM also provides $(\epsilon\sqrt{2\log(1/\delta)/T}+\epsilon(e^{\epsilon/T}-1),\delta)$-differential privacy, resulting in a different optimal $T$.

\begin{corollary}\label{thm:main2}
For any $B$, $Q$, and $\epsilon' > 0, \delta > 0$, there exist $\epsilon > 0$ and $T$ so that MWEM provides $(\epsilon',\delta)$-differential privacy and with probability at least $1-1/poly(|Q|)$ produces $A$ such that
\[
\max_{q \in Q} |q(A) - q(B)|  \textrm{ is } 
O\left( n^{1/2}\left(
\frac{\log^{1/2} |D| \log |Q| \log(1/\delta)}{\epsilon' }\right)^{1/2}\right)\mper
\]
\end{corollary}

Corollary~\ref{thm:main2} is tight in its dependence on $n$ and $|Q|.$
Indeed, there is a lower bound of $\Omega(\sqrt{n\log|Q|})$ on the error of
any algorithm avoiding what is called ``blatant non-privacy"~\cite{DN03,DN04,DY08}. As
$(\epsilon,\delta)$-differential privacy avoids blatant non-privacy for
sufficiently small but constant $\epsilon,\delta,$ the bound applies. We
remark that the lower bound coincides with what is known as the
\emph{statistical sampling error} that arises already in the absence of any
privacy concerns.

Note that these bounds are worst-case bounds, over adversarially
chosen data and query sets. We will see in Section \ref{sec:experimentation}
that MWEM works very well in more realistic settings.

\subsubsection{Running time}
The running time of our basic algorithm as described in
Figure~\ref{fig:algorithm} is $O(n |Q| + T |D| |Q|)).$
The algorithm is embarrassingly parallel: query evaluation can be conducted independently, implemented using modern database technology; the only required serialization is that the $T$ steps must proceed in sequence, but within each step essentially all work is parallelizable.

Results of
Dwork et al.~\cite{DNRRV09} show that for worst case data, producing differentially private
synthetic data for a set of counting queries requires time $|D|^{0.99}$ under
reasonable cryptographic hardness assumptions. Moreover, Ullman and
Vadhan~\cite{UllmanV11} showed that similar lower bounds also hold for more
basic query classes such as we consider in Section~\ref{sec:contingency}. Despite these hardness results, 
we provide an alternate implementation of our
algorithm in Section~\ref{sec:implementation} and demonstrate that its running time is acceptable on real-world data even in cases where $|D|$
is as large as $2^{77}$, and on simple synthetic input datasets where $|D|$ is as large as $2^{1000}$.

\subsubsection{Improvements and Variations}\label{subsec:variations}

We now highlight several variations on the simple approach presented
above that can lead to noticeably improved performance in practice,
although these optimizations do not improve the theoretical worst case
bounds. We incorporate these variations in our experimental validation below.

To achieve the formal bounds, MWEM returns the average over the $A_i$ considered in each round. However, the proof notes that the relative entropy between the distributions $A_i/n$ and $B/n$ decreases with every iteration in which $s_i(B,q_i)$ is large. Rather than return the average $\avg_{i < T} A_i$ we can return $A_T$. We do this in all of our experiments.

In each iteration our algorithm selects a query to measure based on
the amount of error exhibited between our approximating distribution
and the true data. The selected query is measured, and
corrected. However, there is no harm in applying the Multiplicative
Weights update rule multiple times, using all previously taken
measurements. As long as any previously measured $(q_j, m_j)$ pair
is poorly reflected by $A_i$, determined from $|q_j(A_i) - m_j|$ without
reinterrogating $B$, we can improve $A_i$ by reapplying the
Multiplicative Weights step. We do this in all of our experiments. 

Absent priors over the underlying private distribution, our best guess
at the outset is of a uniform distribution, and we initialize our
candidate output distribution accordingly. The quality of this
approximation can sometimes be substantially improved by taking a
histogram over the domain $D$; by simply counting (with noise) the
number of occurrences of each type of record, we can identify values
that occur with substantial frequency and initialize $A_0$
accordingly. This works well if there are several values with high
frequency, but it does consume from the privacy budget, reducing the
accuracy allowed in the query measurement stage. We consider this optimization in Section \ref{sec:contingency}.

For a fixed privacy budget $\epsilon$, the number of iterations $T$ to conduct is the remaining important parameter. Setting it too low results in not enough information extracted about the data, but setting it too high causes each iteration to give very noisy measurements, of little value. Instead, we can set the number adaptively, by starting with a very small $\epsilon$ and asking queries until the observed signal drops below noise levels (or until our budget is expended). If privacy
budget still remains, we double $\epsilon$ and restart. As $\epsilon$ increases we
will only take more measurements, each iteration asking at least as many questions as the last at twice the privacy cost, causing the cumulative cost to telescope and
be within a factor of two of the final cost. In the final run, a value of $T$ is used that is within a factor of two of optimal. We do not apply this parameter selection in our experiments, instead setting the value of $T$ manually.


\subsection{Related Work}

The study of differentially private synthetic data release mechanisms for
arbitrary counting queries began with the work of Blum, Ligett, and
Roth~\cite{BLR08}, who gave a computationally inefficient (superpolynomial in
$|D|$) $\eps$-differentially private algorithm that achieves error that scales only logarithmically with the
number of queries. The dependence on $n$ and $|Q|$ achieved by their algorithm
is $O(n^{2/3}\log^{1/3}|Q|)$ (which is the same dependence we
achieve in
Corollary~\ref{thm:main1}).  Since~\cite{BLR08}, subsequent
work~\cite{DNRRV09,DRV10,RR10,HR10} has focused on computationally more
efficient algorithms (i.e., polynomial in $|D|$) as well as algorithms that
work in the interactive query setting.\footnote{In this setting, the algorithm receives
queries one at a time and answers each before receiving the next.
Subsequent queries may be chosen adaptively based on previous answers.} The
latest of these results is the private Multiplicative Weights method of Hardt and 
Rothblum~\cite{HR10} which achieves error rates of
$O(\sqrt{n\log(|Q|)})$ for $(\eps, \delta)$-differential privacy (which is the same
dependence we achieve, in Corollary~\ref{thm:main2}). While their algorithm works in the
interactive setting, it can also be used non-interactively to produce
synthetic data, albeit at a computational overhead of $O(n).$
MWEM can also be
cast as an instance of a more general Multiplicative-Weights based framework of 
Gupta et al.~\cite{GuptaHRU11}, though our specific instantiation and its practical
appeal were not anticipated in their work.

Barak et al.~\cite{BCD+07} were the first to address the problem of
generating synthetic databases that preserve differential privacy. Their
algorithm maintains utility with respect to a set of contingency
tables (see Section~\ref{sec:contingency} for definitions).
Their approach identifies the complete set of measurements required to reproduce a contingency
table, and takes each one with a uniform level of accuracy. This may make a large
number of redundant or uninformative measurements at the expense of accuracy
in the more interesting queries. Fienberg et al.~\cite{FRY10} observe that, on
realistic data sets, the Barak et al.~algorithm must add so much noise to
preserve differential privacy that the resulting data are no longer
useful. In Section~\ref{sec:contingency}, we improve upon these
results.

Li et al.~\cite{LHR+10} investigated another approach to answering sets of
counting queries. MWEM, when applied to the specific query sets for
which they state results, reduces the formal error bounds substantially. As an example, in the case where $Q$ corresponds to the set of all
$(0,1)$-counting queries, so that $|Q| = 2^{|D|}$, the dependence on $|D|$ goes from $O(|D| \log^2
|D|)$ to $O(|D|^{1/3} \log^{1/3} |D|)$. In Section~\ref{sec:range} we turn to an
empirical comparison of our algorithm with~\cite{LHR+10} and additional related
work on range queries~\cite{XWG09,HRMS10,LHR+10,LiMi11,LiMiAdaptive12,LiMiMeasuring12}. Li and
Miklau~\cite{LiMi11,LiMiMeasuring12} show that these can all be seen as instances of the
\emph{matrix mechanism}. Empirically, we
demonstrate that our approach achieves lower error than a known
\emph{lower bound} on the matrix mechanism.

Ding et al.~\cite{DingWHL11} studied the problem of privately releasing data
cubes.
Their approach is based on various
optimization problems designed to select a set of measurements (i.e., cuboids) on
the data set sufficient to reconstruct a possibly larger target set of
cuboids. When we are interested in all $k$-dimensional cuboids over a
$d$-dimensional binary data set, it can be shown that at least $\binom
dk\approx d^k$ measurements are needed by the Ding et al.~algorithm to ensure bounded error on all cells.
In contrast, MWEM would require by virtue of Corollary~\ref{thm:main1}
less than $d^{1/3}\cdot n^{2/3}$ measurements (update iterations) for the
same purpose. Even for modest $k$ and $|D|$ this results in
significantly fewer measurements and thus smaller error. In
Section~\ref{sec:datacubes} we empirically compare our algorithm to the work
of Ding et al.

\section{Experimental Evaluation}
\label{sec:experimentation}

We evaluate MWEM across a variety of query classes, datasets, and
metrics as explored by prior work, demonstrating improvement in the quality of approximation (often
significant) in each case. The problems we consider are: (1) range
queries under the total squared error metric, (2) binary contingency table
release under the relative entropy metric, and (3) datacube release under
the average absolute error metric. Although contingency table release
and datacube release are very similar, prior work on the two have had
different focuses: small datasets over binary attributes vs.\ large
datasets over categorical attributes, low-order marginals vs.\ all
cuboids, and relative entropy vs.\ the average error within a
cuboid as metrics. We do not report on running times in this
section. Instead, Section \ref{sec:implementation} describes an
optimized implementation and evaluates it at scale. 

All of our experiments are done with $\epsilon$-differential privacy;
that is, $\delta = 0$. Our $(\epsilon', \delta)$-differential privacy
results are achieved by recharacterizing an $\epsilon$-differentially private execution with different privacy parameters. While the absolute
numbers in the privacy-utility trade-off may improve if we
recharacterize them using a non-zero $\delta$, 
it is not necessary to conduct separate experiments for this case.

\begin{figure*}[htbp]
\begin{center}
\includegraphics[height=4.5cm]{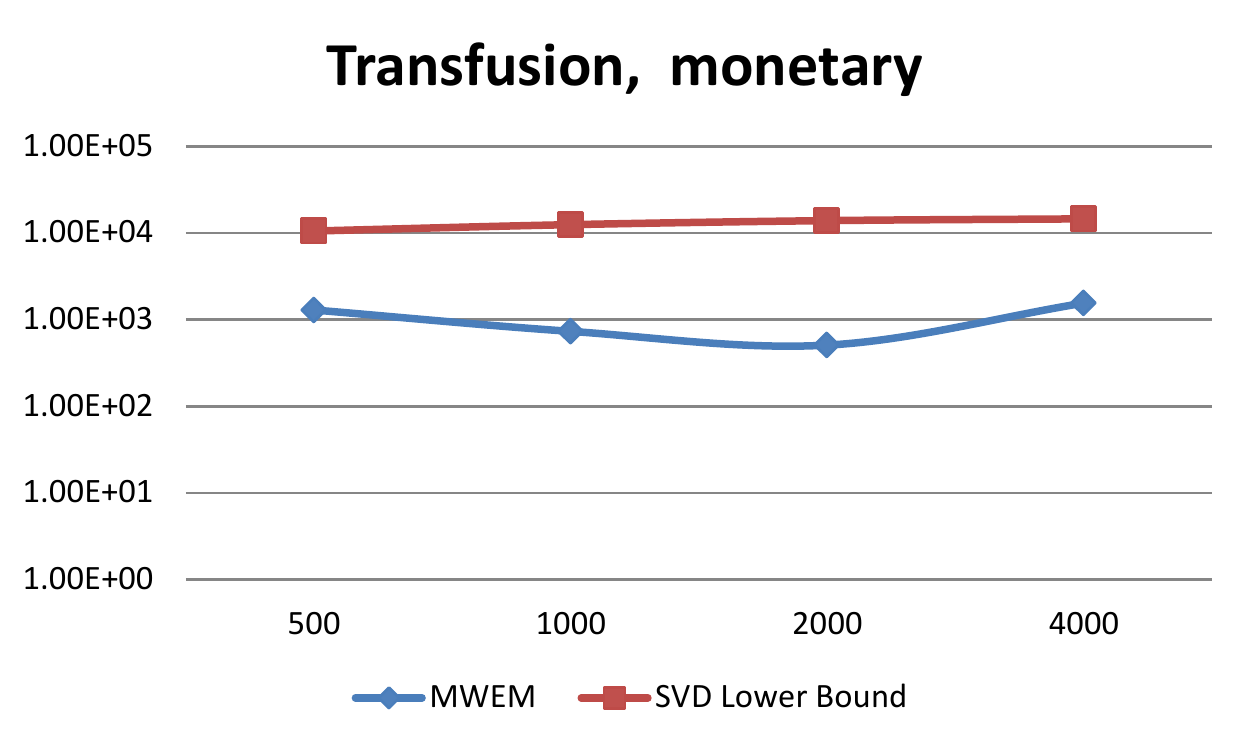}
\includegraphics[height=4.5cm]{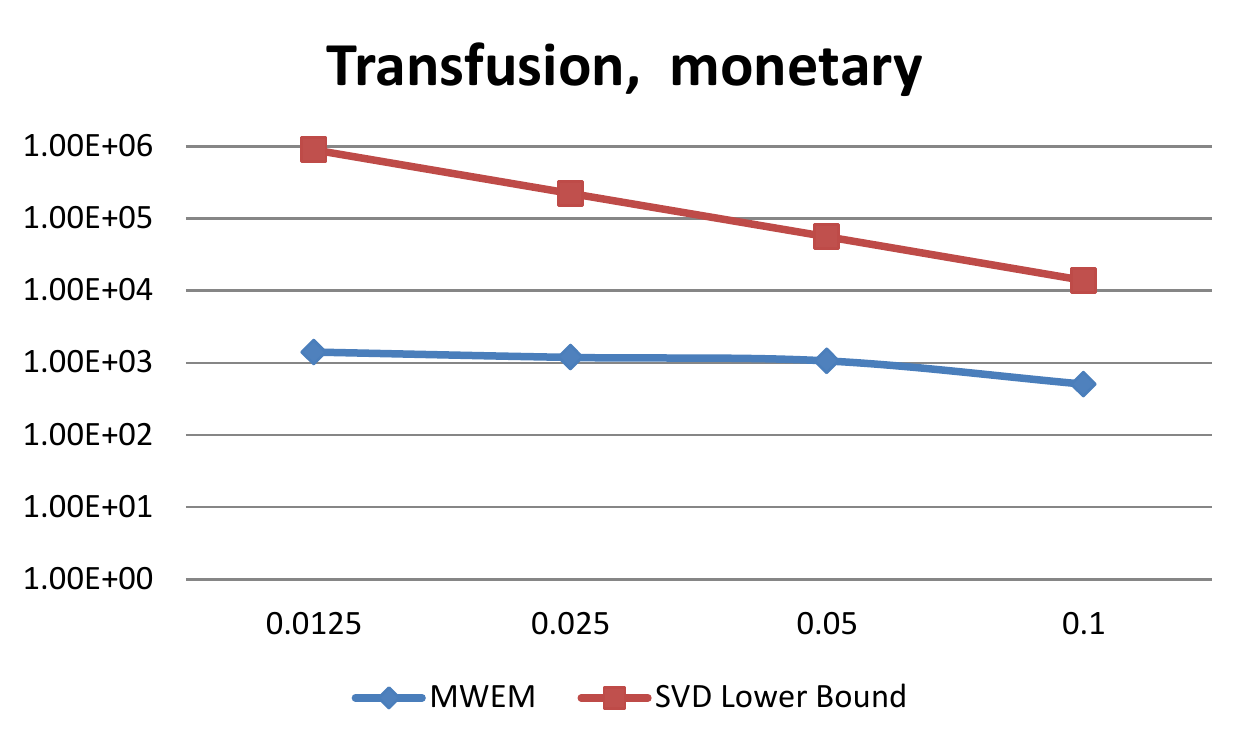}\\
\includegraphics[height=4.5cm]{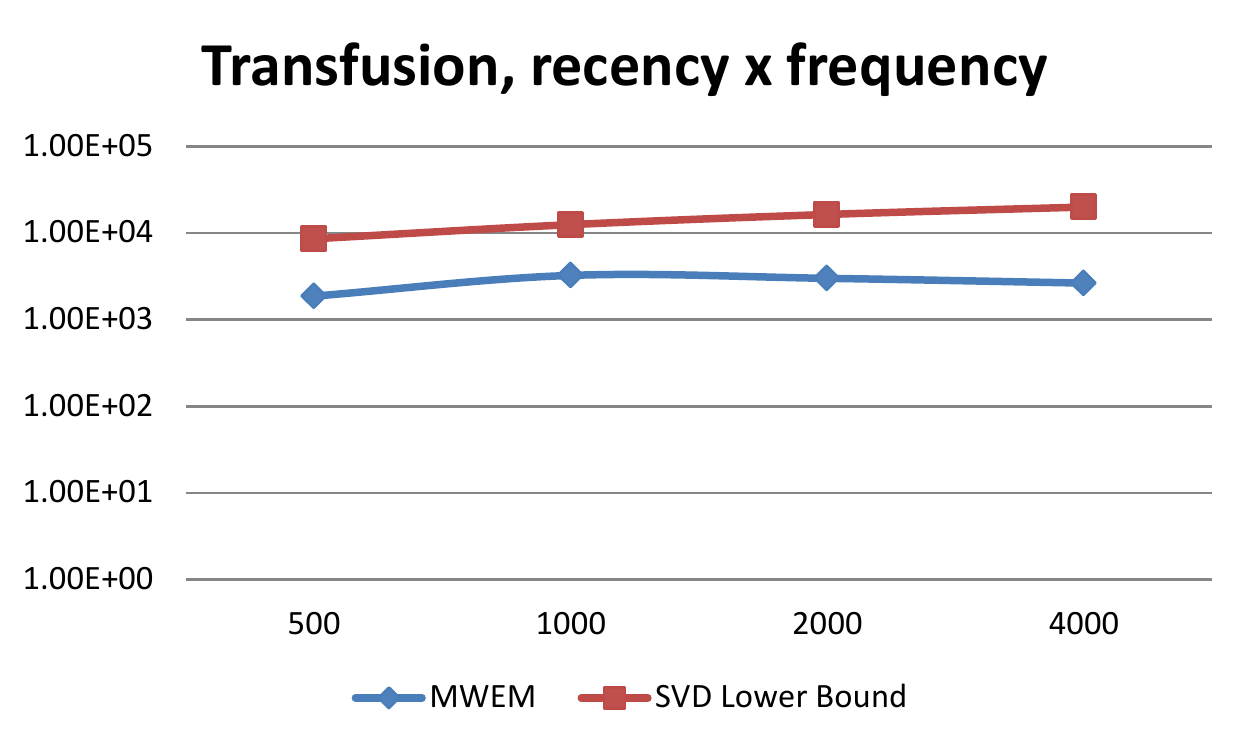}
\includegraphics[height=4.5cm]{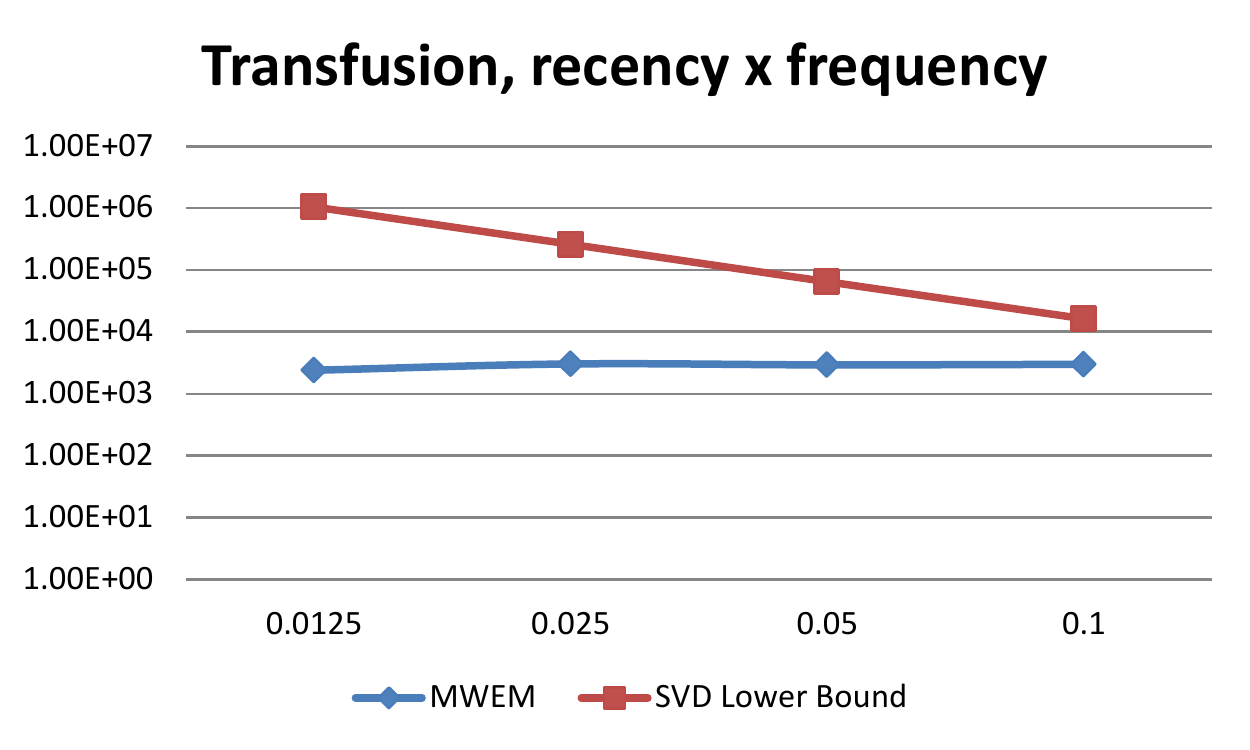}\\
\includegraphics[height=4.5cm]{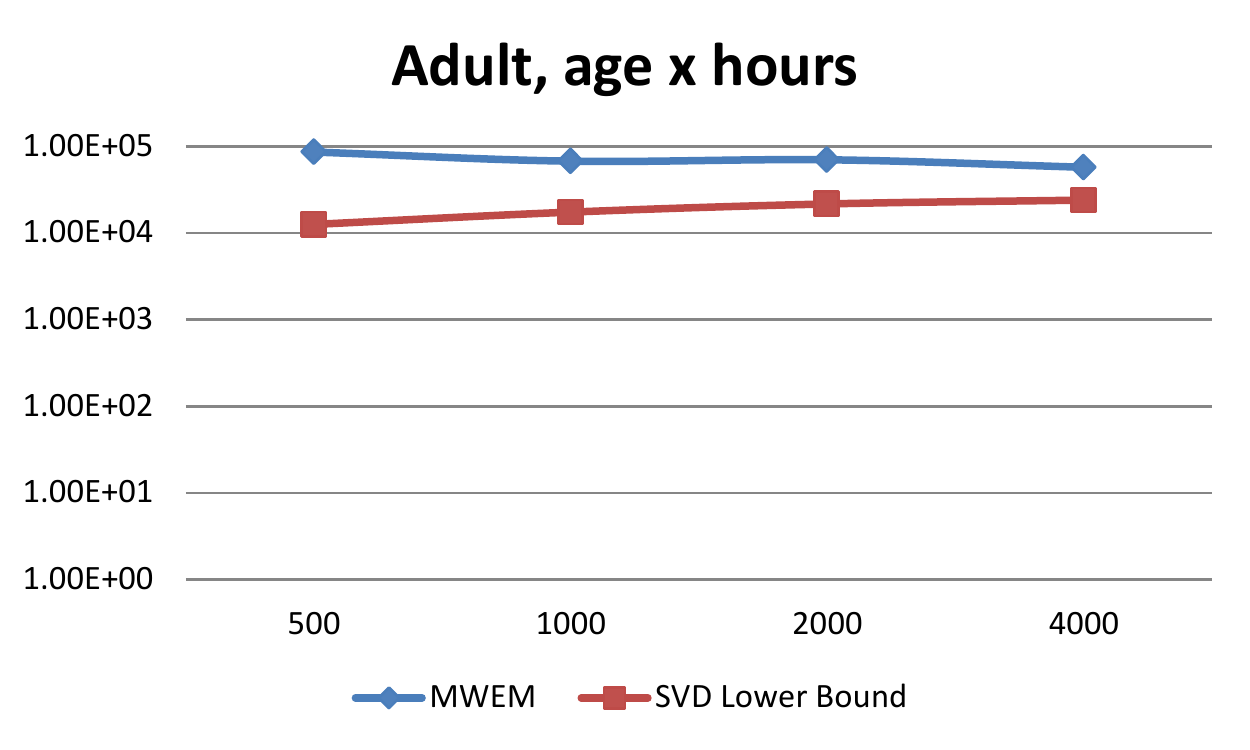}
\includegraphics[height=4.5cm]{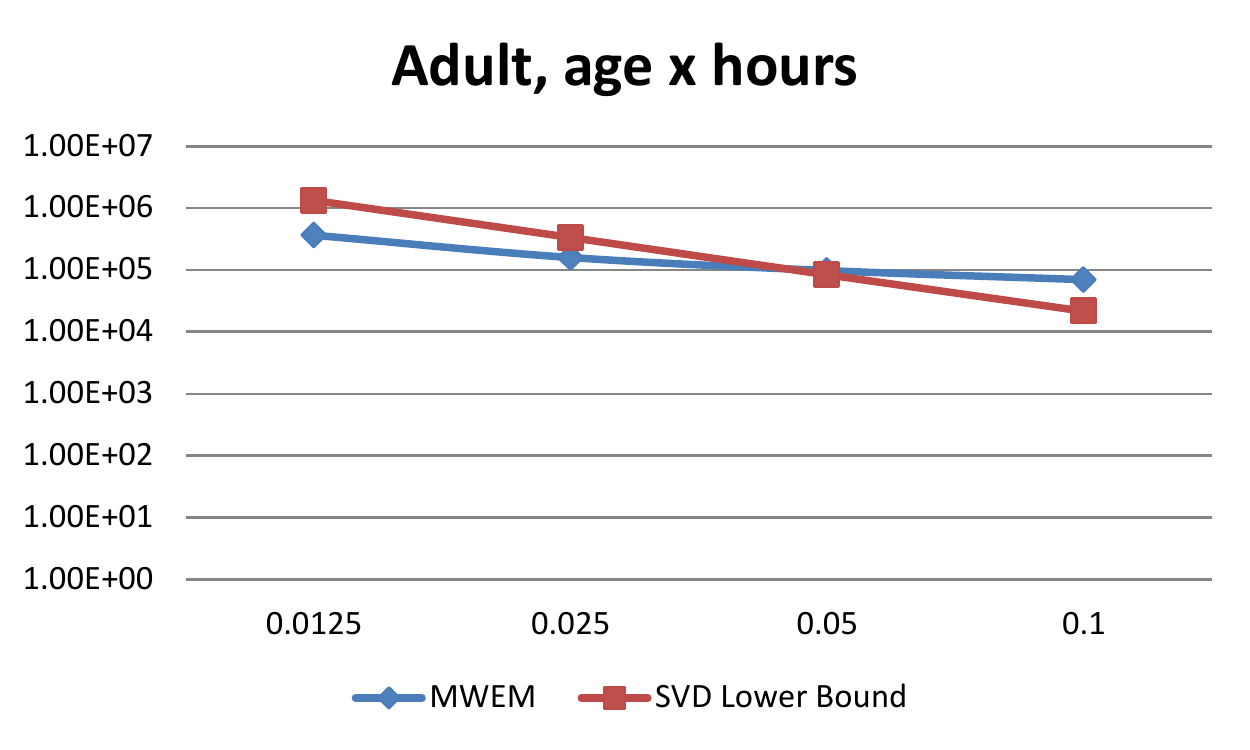}
\includegraphics[height=4.5cm]{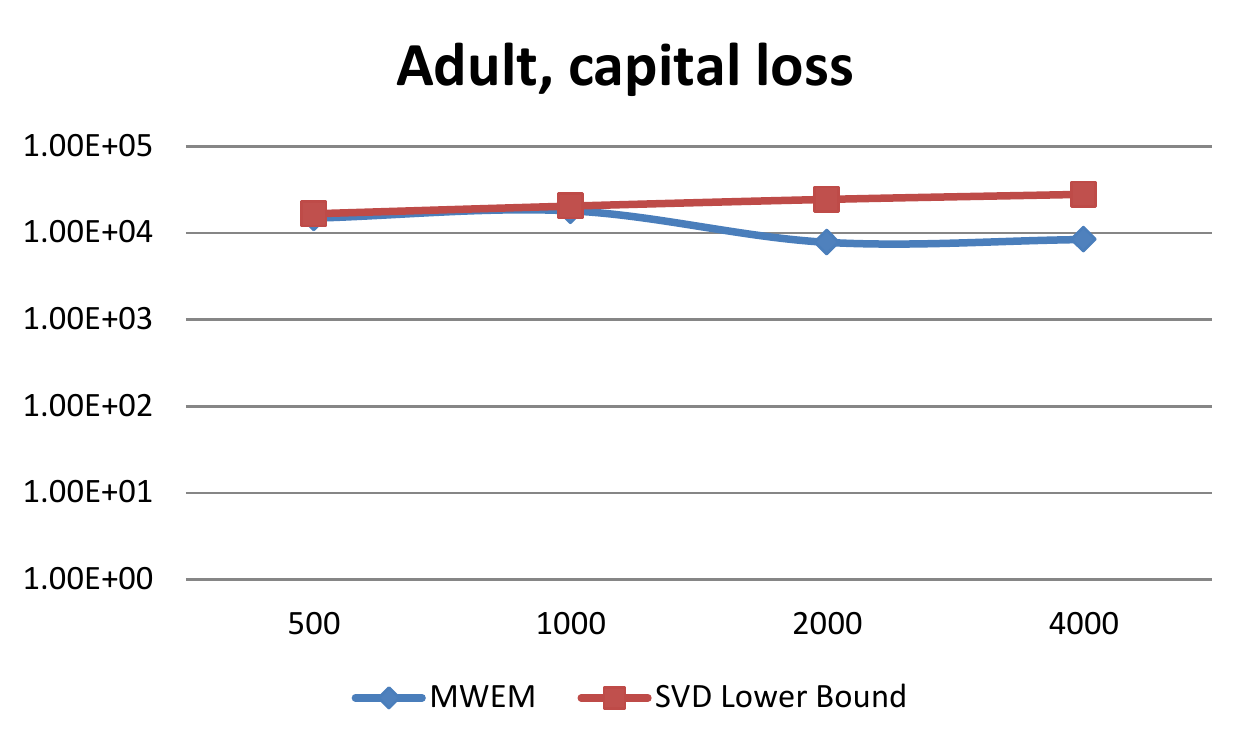}
\includegraphics[height=4.5cm]{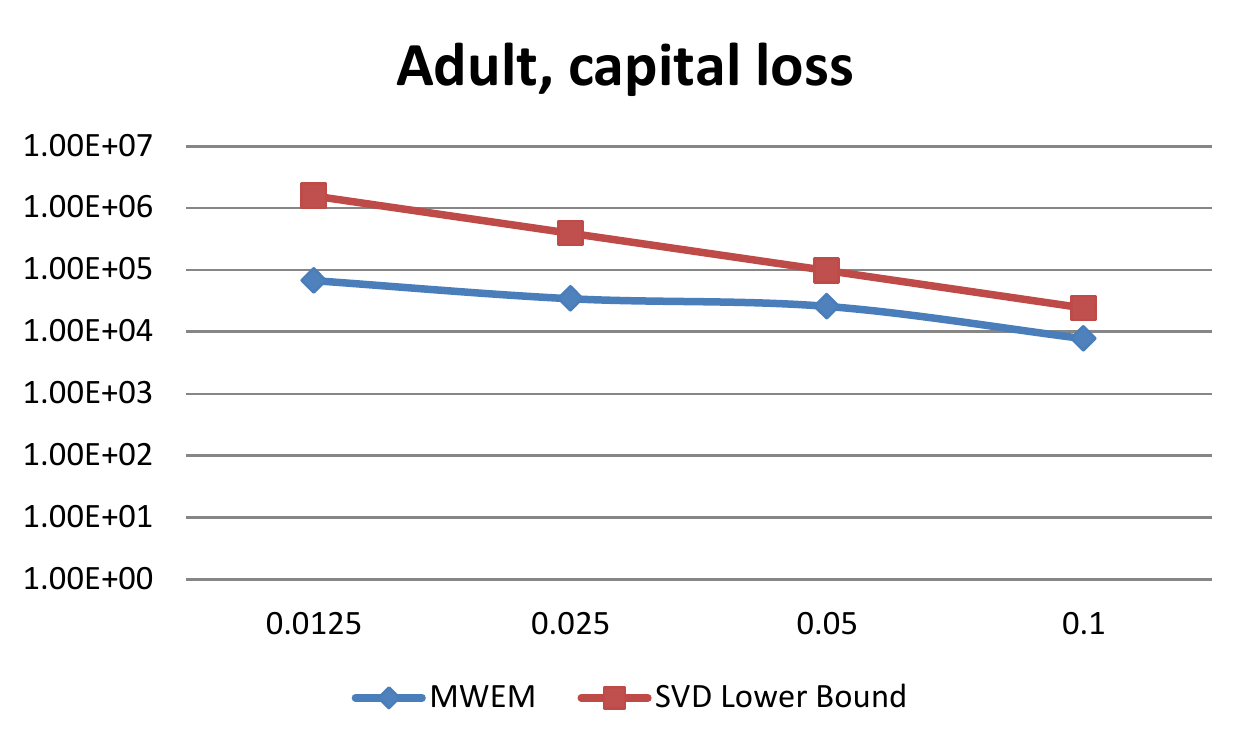}
\caption{Comparison of MWEM with the SVD lower bound on four data
sets. The $y$-axis in each plot represents the average squared error per query.
On the left hand side we vary the number of queries from $500$ to $4000$ while
keeping $\epsilon=0.1$ fixed. On the
right hand side, we vary $\epsilon$ from $0.0125$ to $0.1$ while keeping the
number of queries $|Q| = 2000$ fixed. 
We report the
average over $5$ independent repetitions of the experiment.
On
each data set for sufficiently small $\epsilon$ the error achieved by our
algorithm is lower than the SVD bound. 
}
\label{fig:range}
\end{center}
\end{figure*}

\subsection{Range Queries}
\label{sec:range}

A range query over a domain $D=\{1,\dots,N\}$ is a counting query specified 
by the indicator function of an interval $I\subseteq D.$ The concept extends
naturally to multi-dimensional domains $D=D_1\times \dots D_d$ where
$D_i=\{1,\dots,N_i\}.$ Here, a range query is defined by the indicator
function of a cross product of intervals: the function value is $1$ if and only if each coordinate lies in the associated interval.

\begin{figure*}[htbp] \begin{center}
\includegraphics[height=4cm]{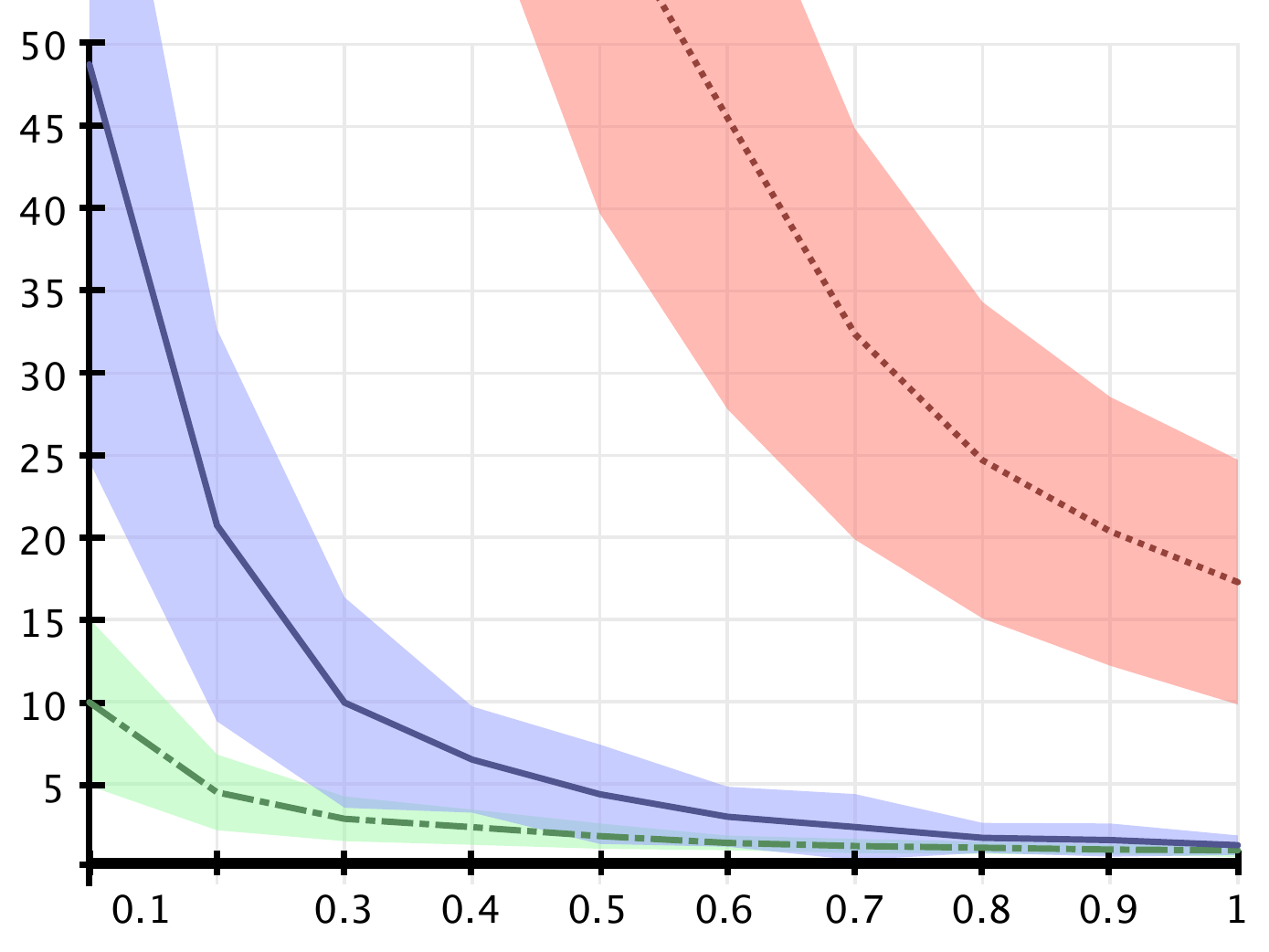}
\includegraphics[height=4cm]{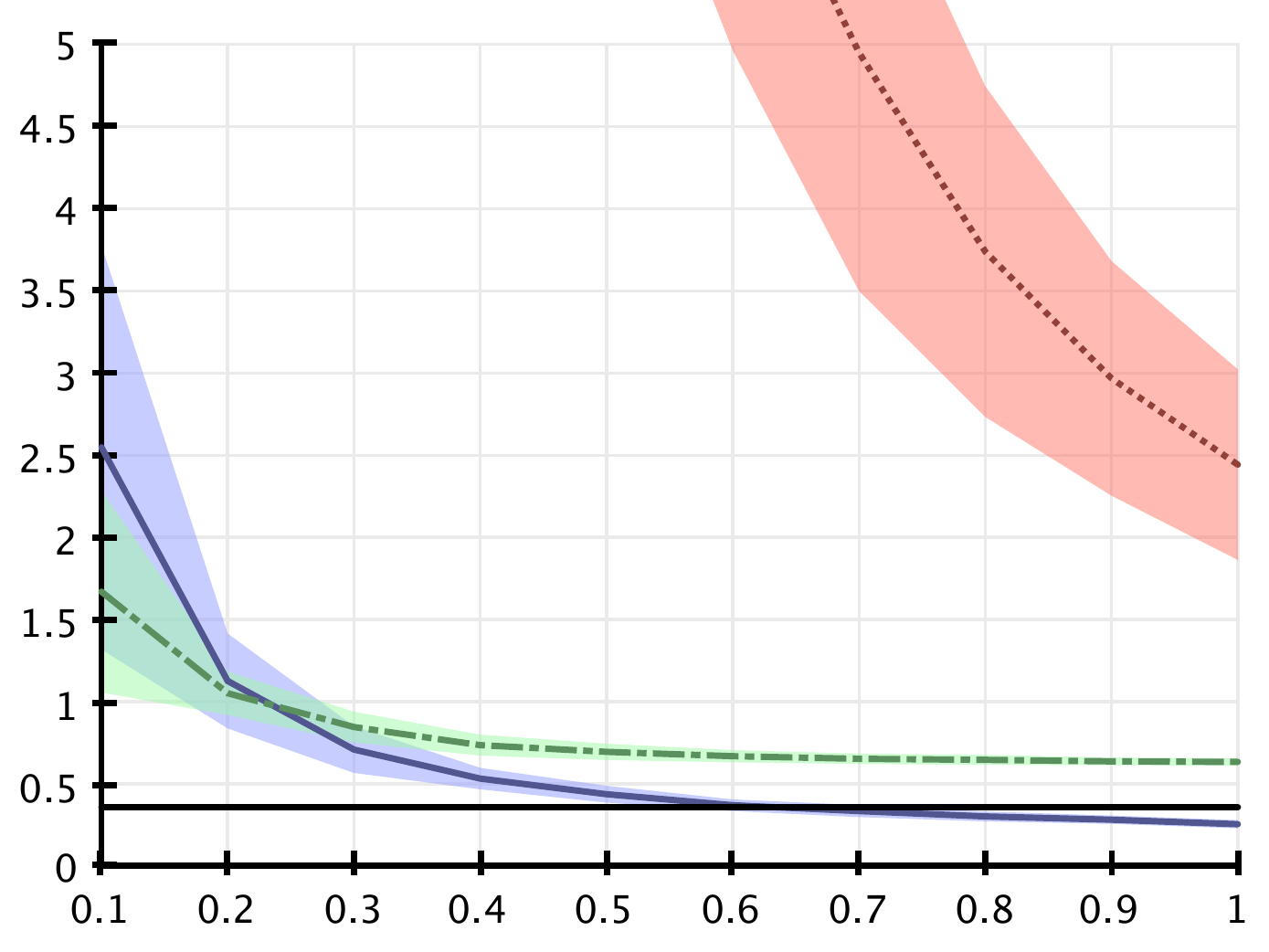}
\includegraphics[height=4cm]{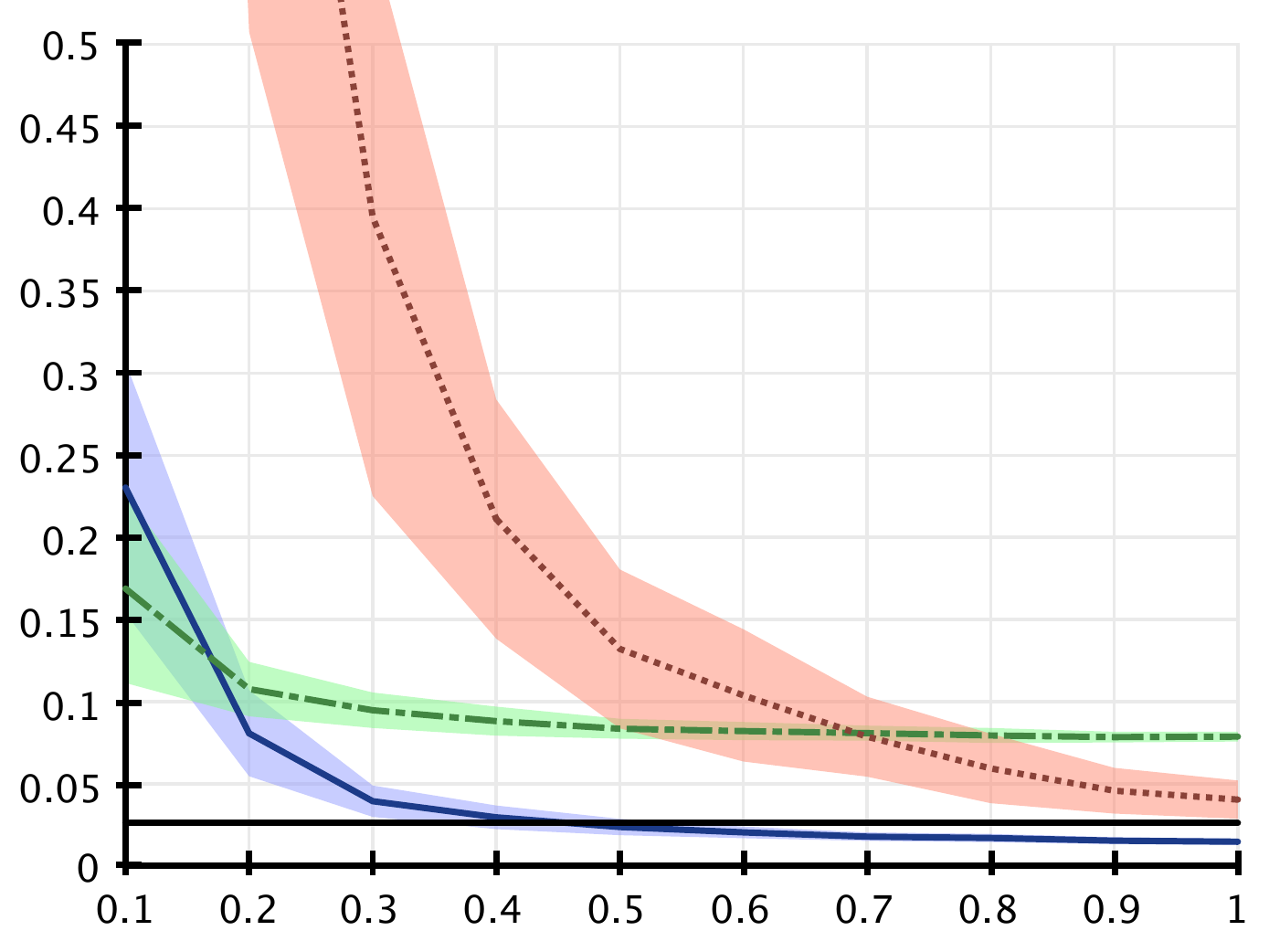}
\caption{Curves describing the behavior of algorithms on the mildew, rochdale,
and czech datasets, respectively. The x-axis is the value of epsilon
guaranteed, and the y-axis is the relative entropy between the produced
distribution and actual dataset. The lines represent averages across 100 runs,
and the corresponding shaded areas one standard deviation in each direction.
Red (dashed) represents the modified Barak et al.~algorithm, green
(dot-dashed) represents unoptimized MWEM, and blue (solid)
represents the optimized version thereof. The solid black horizontal line is
the stated relative entropy values from Fienberg et al.} \label{smalldatasets}
\end{center} \end{figure*}

Differentially private algorithms for range queries were specifically
considered by \cite{BLR08,XWG09,HRMS10,LHR+10,LiMi11,LiMiAdaptive12,LiMiMeasuring12}. As noted
in~\cite{LiMi11,LiMiMeasuring12}, all previously implemented algorithms for range
queries can be seen as instances of the \emph{matrix mechanism}
of~\cite{LHR+10}. Moreover, \cite{LiMi11,LiMiMeasuring12} show a
\emph{lower bound} on the total squared error achieved by the matrix
mechanism in terms of the singular values of a matrix
associated with the set of queries. We refer to this bound as the
\emph{SVD bound}.

We empirically evaluate MWEM for range queries on several
real-world data sets.  Specifically, we consider (a) the ``capital
loss'' attribute of the Adult data set~\cite{uci} corresponding to a
domain of size $4357,$ (b) the combined ``age'' and ``hours''
attributes of the Adult data set corresponding to a domain of size
$91\times 25$, (c) the combined ``recency'' and ``frequency''
attributes of the Blood Transfusion data set~\cite{uci,blood} corresponding
to a domain of size $80\times 55$, and (d) the ``monetary'' attribute of
the Blood Transfusion data set corresponding to a domain of size
$1251.$ We chose these data sets as they feature numerical attributes
of suitable
size. 

\subsubsection{Experimental results}
In Figure~\ref{fig:range}, we compare the
performance of MWEM on sets of randomly chosen range queries
against the SVD lower bound proved by~\cite{LiMi11,LiMiMeasuring12}, (1) varying the
numbers of queries ($|Q|$) while keeping $\epsilon$ fixed, and (2) varying
$\epsilon$ while keeping the number of queries fixed. We chose
$T\in\{10,12,14,16\}$ for our experiments and reported the values for the best
setting of $T$ in each case. The reported numbers are averages of $5$
independent repetitions of the same experiment. We report the total squared
error and SVD bound normalized by the total number of queries. 

In all four cases we observe that the
error achieved by MWEM is \emph{lower} than the SVD lower bound on the
matrix mechanism, for sufficiently small privacy parameter~$\epsilon.$ The
improvement against the SVD bound is often by more than an order of
magnitude and sometimes by up to three orders of magnitude. 
Moreover, the privacy guarantee
achieved by our algorithm is $(\epsilon,0)$-differential privacy whereas the
SVD lower bound holds for algorithms achieving the strictly weaker
guarantee of $(\epsilon,\delta)$-differential privacy with $\delta>0.$ The SVD bound depends on $\delta$; in our experiments we
fixed $\delta=1/n$ when instantiating the SVD bound, as any larger value of $\delta$ would permit exact release of individual records.


\subsection{Contingency Tables}
\label{sec:contingency}

A contingency table can be thought of as a table of records over $d$
binary attributes, and the $k$-way marginals of a contingency table
correspond to the $d \choose k$ possible choices of $k$ attributes,
where each marginal is represented by the $2^k$ counts of the records
with each possible setting of attributes. Marginals of contingency tables
exhibit interesting correlations between queries, and have a
significant role in the practice of official statistics.  When
statistical inference is performed over contingency tables, data
analysts seek sets of low-dimensional marginals (i.e., containing
relatively few attributes at a time) that fit the data well. Our goal in
releasing a differentially private contingency table is to preserve
these low-dimensional marginals.

In previous work, Barak et al.~\cite{BCD+07} describe an approach to
differentially private contingency table release through the Hadamard
transformation. If we view a contingency table as vector with the
coordinates ordered lexicographically by (binary) attribute settings,
the Hadamard transformation corresponds to multiplication by the
Hadamard matrix, defined recursively as
\[ 
H_{n+1} = \left[\begin{array}{rr} H_n & H_n 
\\ H_n & -H_n \end{array}
\right]\mcom 
\qquad H_1 = [1]\mper 
\] 
Importantly, all $k$-dimensional marginals can be exactly recovered by
examination of relatively few entries in the transformed vector (roughly
$d \choose k$ out of $2^d$). Each of these entries corresponds to a set
of at most $k$ out of $d$ attributes, where the entry is the difference between
the number of records with an even number of bits set and the number of
records with an odd number of bits set. Rather than explain which
entries to combine and how (details can be found in Barak et 
al.~\cite{BCD+07}) we simply take the measurements as our query set,
thereby ensuring that our output will respect them. The marginals can then be derived directly from the
dataset we produce.
Note that we do not need to specify the relationship
  between the measurements we take and the quantities of interest; we
  only need that the relationships exist. This is helpful in settings where the
  dependence is complicated or inexact.

\pagebreak

\subsubsection{Experimental Setup}

In this section, we consider several datasets used in the statistical
literature. The datasets, detailed in
Table~\ref{datasets}, range from small (70 records) to
substantial (21k records).

\begin{table}[tdp] \begin{center} \begin{tabular}{l|c|c|c|} & records &
attributes & non-zero / total cells\\ \hline mildew & 70	 & 6 & 22 /
64\\ \hline czech & 1841 & 6 & 63 / 64\\ \hline rochdale & 665 & 8 & 91 /
256\\ \hline nltcs & 21574 & 16 & 3152 / 65536  \\ \hline \end{tabular}
\caption{Details of the four datasets we consider.}
\label{datasets} \end{center} \end{table}

We evaluate our approximate dataset with the truth using \emph{relative entropy}, also known as the Kullback-Leibler (or KL) divergence. Formally, the relative entropy between our two distributions ($A/n$ and $B/n$) is
$$ RE(B||A) = \sum_{x \in D} B(x) \log (B(x) / A(x)) / n\; .$$
This measurement has
appealing properties for statistical inference, and is used in previous
statistical work on the problem. 
%
Our experiments are intended both to compare our approach to the prior work of
Barak et al.~\cite{BCD+07} as well as to evaluate it in absolute terms. For
the purposes of our experiments, Barak et al.~is represented by
the approach that takes all low order measurements with a uniform
level of accuracy and applies Multiplicative Weights (the approach in
\cite{BCD+07} involved a linear programming 
step instead of Multiplicative Weights, which we have found only hurts
its performance with respect to relative entropy). 
For the absolute comparison, we invoke the work of Fienberg et
al.~\cite{FRY10} on several of these datasets where they report absolute
numbers for quality of fit (in terms of relative entropy) without privacy
constraints, but at a certain level of statistical generality (that is, they do not want to overfit). 

We consider both a standard application of MWEM, and one in which we use
two optimizations described in Section~\ref{subsec:variations}:
re-execution of measured queries and
initialization using a histogram over the domain $D$ of records.

\subsubsection{Experimental Results}

We first evaluate MWEM on several small datasets in common use by
statisticians. Our findings here are fairly uniform across the datasets: the
ability to measure only those queries that are informative about the dataset
results in substantial savings over taking all possible measurements. We
evaluate both our theoretically pure algorithm and its heuristic improvement
as discussed in the previous section, against a modified version of the
algorithm of Barak et al.~\cite{BCD+07} (improved by integrating the Multiplicative
Weights of Hardt-Rothblum~\cite{HR10}), and the accepted ``good" non-private
relative entropy values from Fienberg et al.~\cite{FRY10}. The trade-off
between relative entropy and $\epsilon$ for three datasets appears in Figure
\ref{smalldatasets}.  In each case, we see that we noticeably improve on the
algorithm of Barak et al., and in many cases our heuristic approach matches
the good non-private values of~\cite{FRY10}, indicating that we can
approach levels of accuracy at the limit of statistical validity.

We also consider a larger dataset, the National Long-Term Care Study
(NLTCS), in Figure \ref{largedataset}. This dataset contains orders of
magnitudes more records, and has 16 binary attributes. For our initial
settings, maintaining all three-way marginals, we see
similar behavior as above: the ability to choose the measurements that
are important allows substantially higher accuracy on those that
matter.  However, we see that the algorithm of Barak et
al.~\cite{BCD+07} is substantially more competitive in the regime
where we are interested in querying all two-dimensional marginals,
rather than the default three we have been using. In this case, for
values of epsilon at least $0.1$, it seems that there is enough signal
present to simply measure all corresponding entries of the Hadamard transform; each is
sufficiently informative that measuring substantially fewer at higher
accuracy imparts less information, rather than more.

For every dataset and query set, there is some sufficiently high epsilon level
where the judicious selection of queries is no longer required. In such
regimes, the approach we present in this paper does not provide an improvement
over more naive approaches. The impact of our approach returns if we increase
the dimension of the marginal that must be preserved (dramatically increasing the
number of measurements Barak et al.~would take) or if we decrease~$\epsilon$ to a
level such that the majority of two-way measurements are not above the
noise level, both of which are demonstrated in Figure \ref{largedataset}. However, the analyst's goal should be to get the right output for
the analysis task at hand, under the supplied privacy constraints. In some
cases this may not require our advanced query selection.

\begin{figure*}[htbp]
 \begin{center}
\includegraphics[height=4cm]{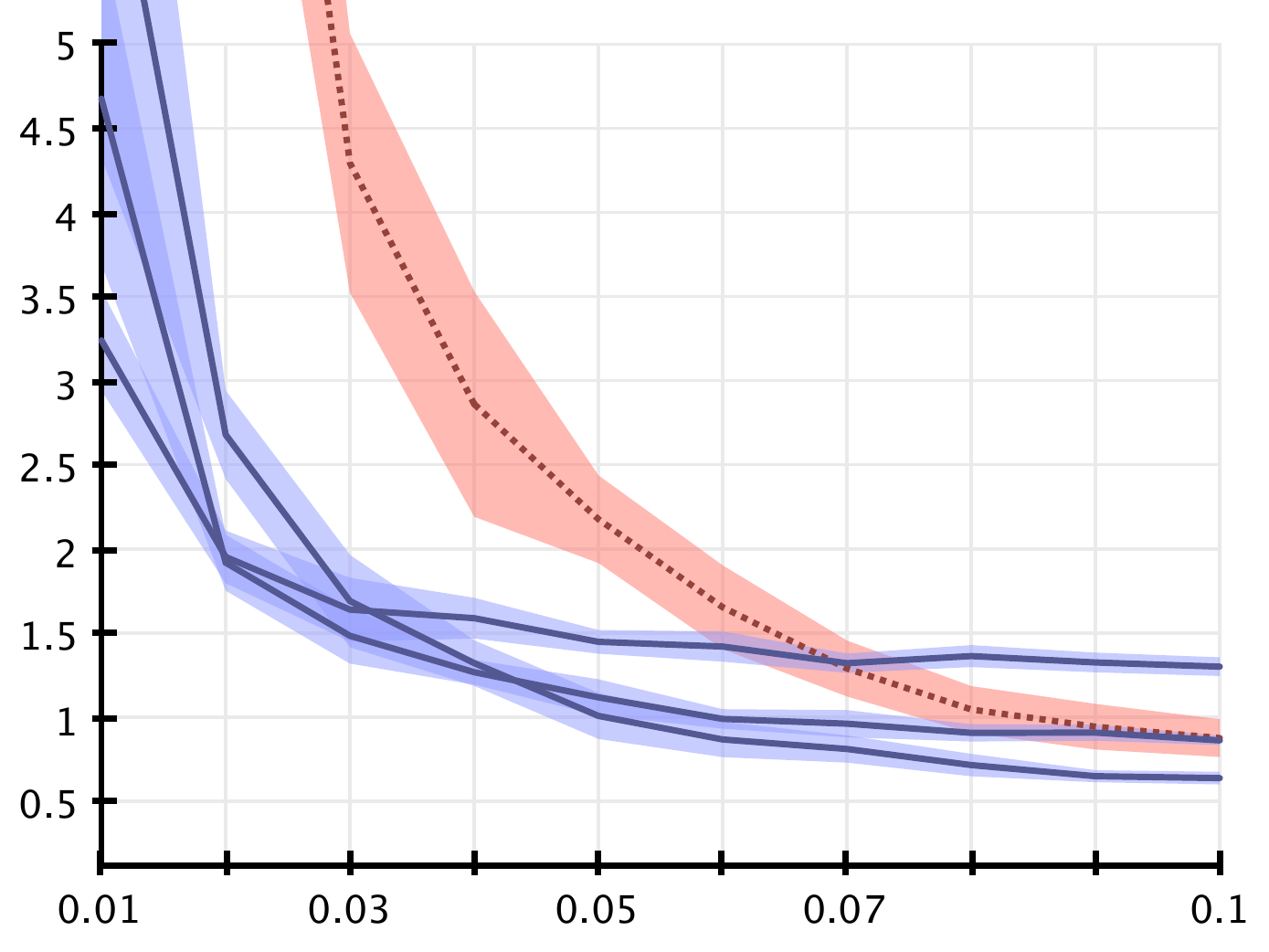}
\includegraphics[height=4cm]{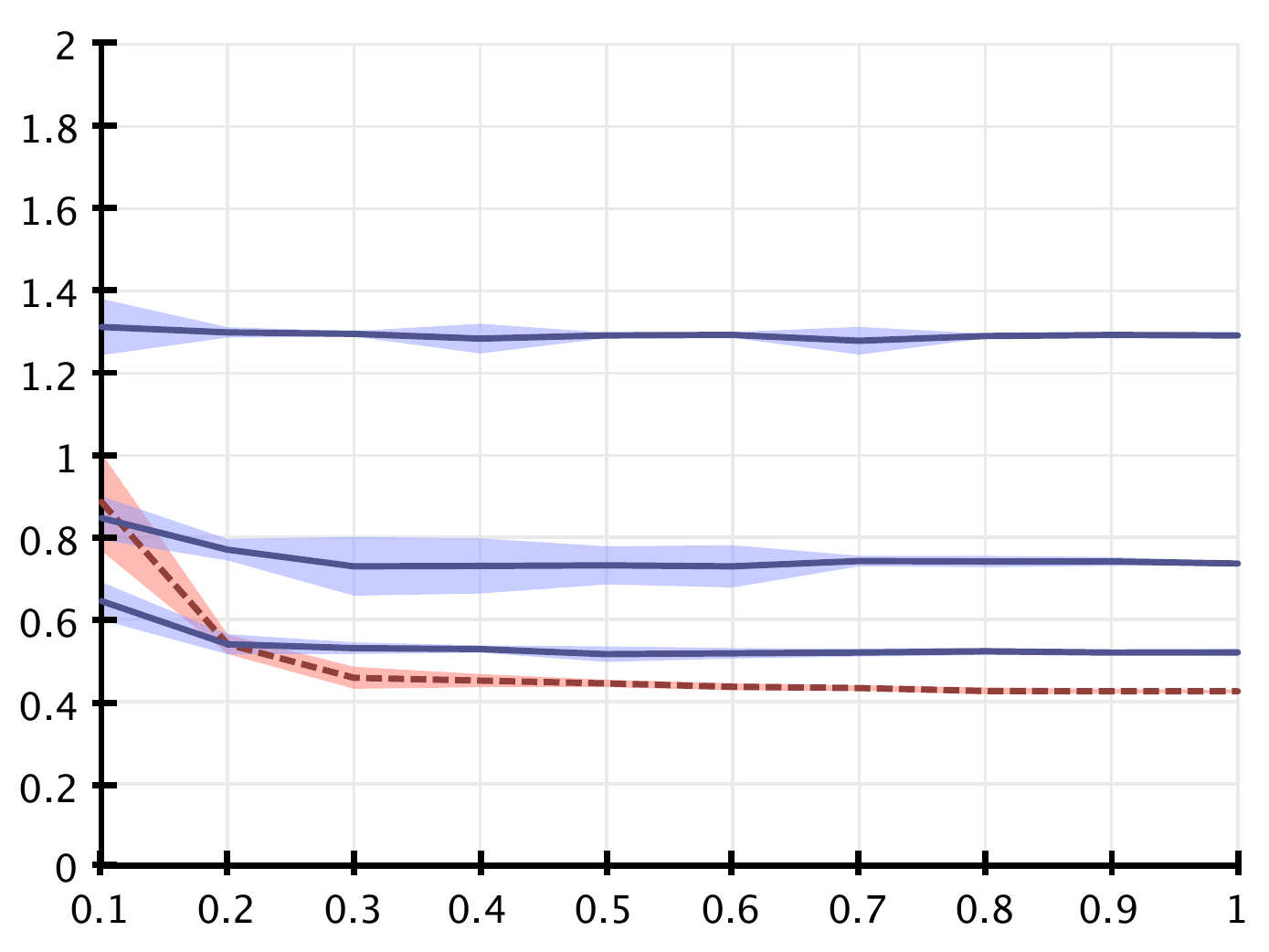}
\includegraphics[height=4cm]{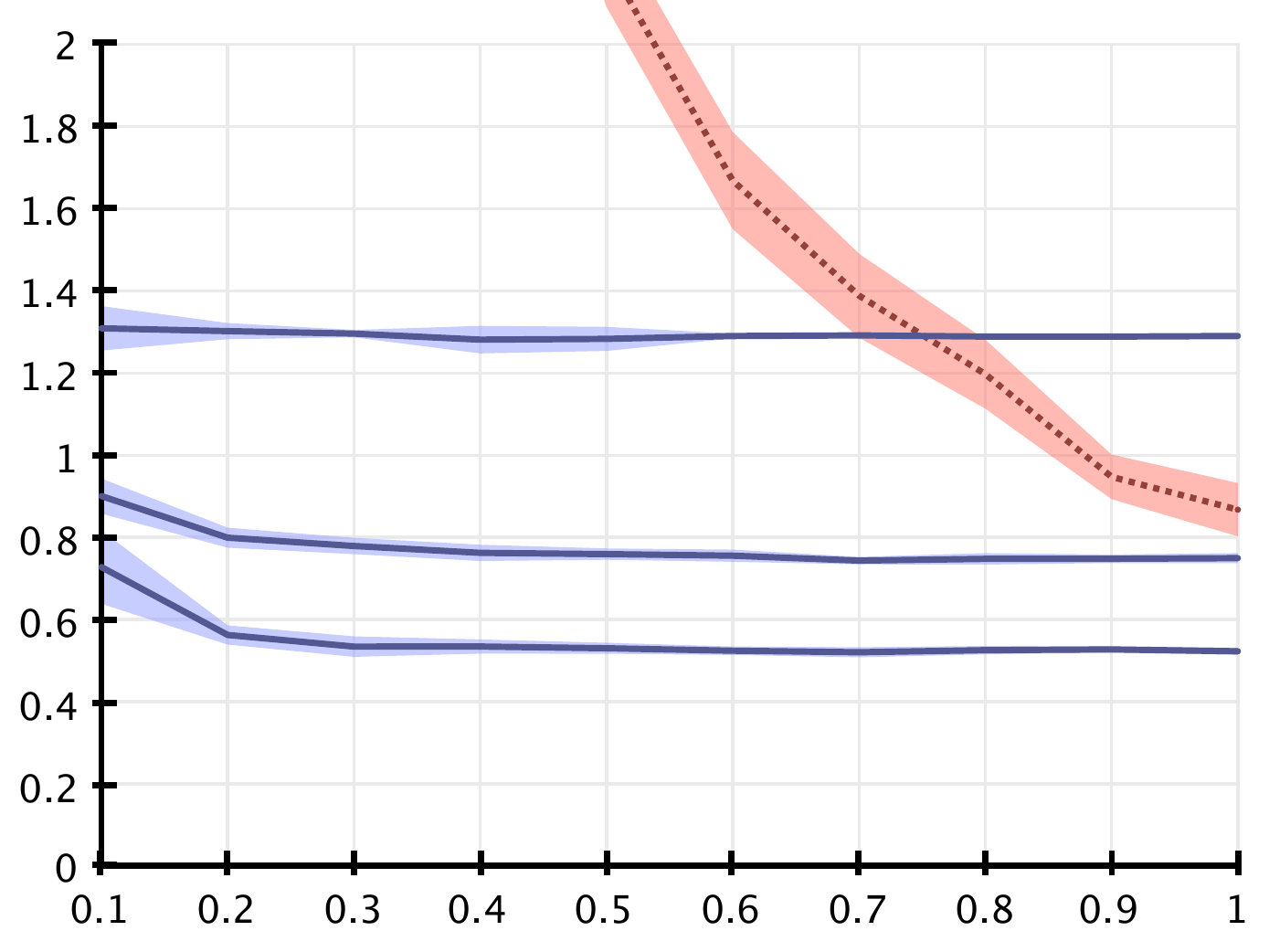}
\caption{Curves comparing our approach with that of Barak et al.~on the
National Long Term Care Survey. The red (dashed) curve represents Barak et al,
and the multiple blue (solid) curves represent MWEM, with 20, 30, and 40 queries (top to bottom,
respectively). From left to right, the first two figures correspond to degree
2 marginals, and the third to degree 3 marginals. 
As
before, the x-axis is the value of epsilon guaranteed, and the y-axis is the
relative entropy between the produced distribution and actual dataset. The
lines represent averages across only 10 runs, owing to the high complexity of
Barak et al.~on this many-attributed dataset, and the corresponding shaded
areas one standard deviation in each direction.\label{largedataset}} 
\end{center} \end{figure*}

\subsection{Data Cubes}
\label{sec:datacubes}
We now change our terminology and objectives, shifting our view of
contingency tables to one of datacubes. The two concepts are
interchangeable, a contingency table corresponding to the datacube, and
a marginal corresponding to its cuboids. However, the datasets studied
and the metrics applied are different. We focus on the restriction of the Adult dataset~\cite{uci} to its
eight categorical attributes, as done in \cite{DingWHL11}, and evaluate
our approximations using average error within a cuboid, also as done in
\cite{DingWHL11}.

Although MWEM is defined with respect to a single query at a time, it
generalizes to sets of counting queries, as reflected in a cuboid. The
Exponential Mechanism can select a cuboid to measure using a quality score
function summing the absolute values of the errors within the cells of the cuboid. We also
(heuristically) subtract the number of cells from the score of a cuboid to
bias the selection away from cuboids with many cells, which would collect Laplace
error in each cell. This subtraction does not affect privacy properties. An entire cuboid can be measured with a single
differentially private query, as any record contributes to at most one cell
(this is a generalization of the Laplace Mechanism to multiple dimensions,
from \cite{DMNS06}). Finally, Multiplicative Weights works unmodified, increasing
and decreasing weights based on the over- or under-estimation of the count to
which the record contributes.

\subsubsection{Experimental Results}

We apply MWEM to the Adult dataset in several ways:
restricting our computation to 2-way cuboids, 3-way cuboids, and with
no restriction. As it turns out, the latter two are identical, in that
the higher order cuboids have too many cells to be appealing to the
algorithm, and are ignored by MWEM. In fact, the 3-way cuboid experiment used
only one 3-way cuboid, albeit a very helpful one. We plot the maximum
cuboid error and average cuboid error in Figure
\ref{fig:cuboids}. Comparing our final 3-way measurements ($\max =
138.71$ and $\textrm{avg} = 13.21$) to the $\epsilon = 1$ reading in
Figure 3 of \cite{DingWHL11}, the maximum error appears comparable to
their best result, whereas the average error appears approximately four
times lower than their best result. Of note, our results are achieved
by a single algorithm, whereas the best results for maximum and
average error in \cite{DingWHL11} are achieved by different
algorithms, each designed to optimize one specific metric.

\section{A Scalable Implementation}
\label{sec:implementation}

The implementation of MWEM used in the previous experiments quite literally
maintains a distribution $A_i$ over the elements of the universe~$D$. 
As the number of attributes grows, the universe~$D$ grows exponentially, 
and it can quickly become infeasible to track the distribution explicitly. 
In this section, we consider a
scalable implementation with essentially no memory footprint, whose
running time is in the worst case proportional to $|D|$, but which for
many classes of simple datasets remains linear in the number of
attributes.

First, recall that the heart of MWEM is to use Multiplicative
Weights to maintain a distribution $A_i$ over $D$ that is then used in
the Exponential Mechanism to select queries poorly approximated by the
current distribution. 
From the definition of the Multiplicative Weights distribution, we see
that the weight $A_i(x)$ can be determined from the history $H_i = \{(q_j,m_j) : j \le i \}$:
\[
A_i(x) \propto \exp\left(\sum_{j \le i} q_j(x) \times (m_j - q_j(A_{j-1}))/2n\right)\mper
\]

We explicitly record the scaling factors $l_j = m_j - q_j(A_{j-1})$
as part of the history $H_i = \{(q_j, m_j, l_j) : j \le i \}$, to remove
the dependence on prior $A_{j}$. If one is willing to 
iterate over all $x \in D$, one can evaluate each query~$q(A_i)$
using only this history. Additionally, the summation over $x \in D$
is extremely parallelizable, and distributes easily across multiple
cores and computers. While the implicit representation of the
distribution $A_i$ in terms of the history represents a substantial savings in
memory footprint, $D$ can still be exponential in the number of
attributes, and even a large cluster is quickly overwhelmed
by the required computation.

\begin{figure}[tbp] 
\begin{center}
\includegraphics[height=4.5cm]{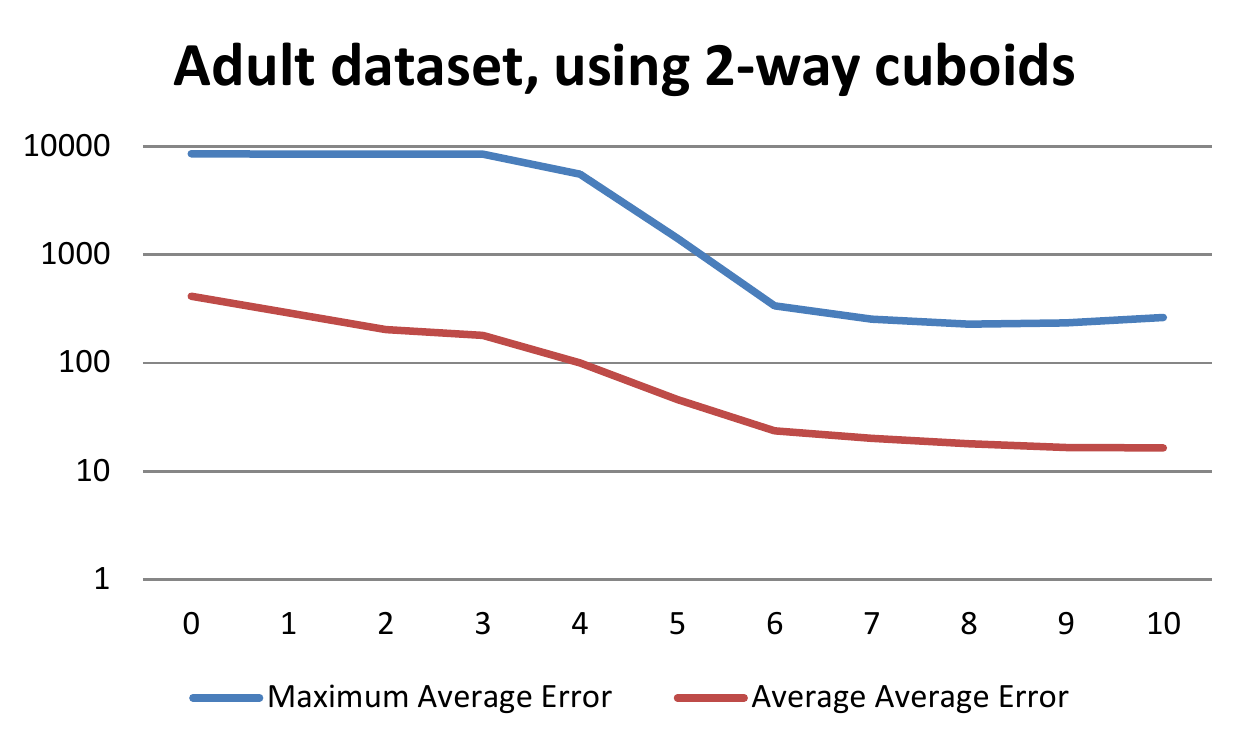}
\includegraphics[height=4.5cm]{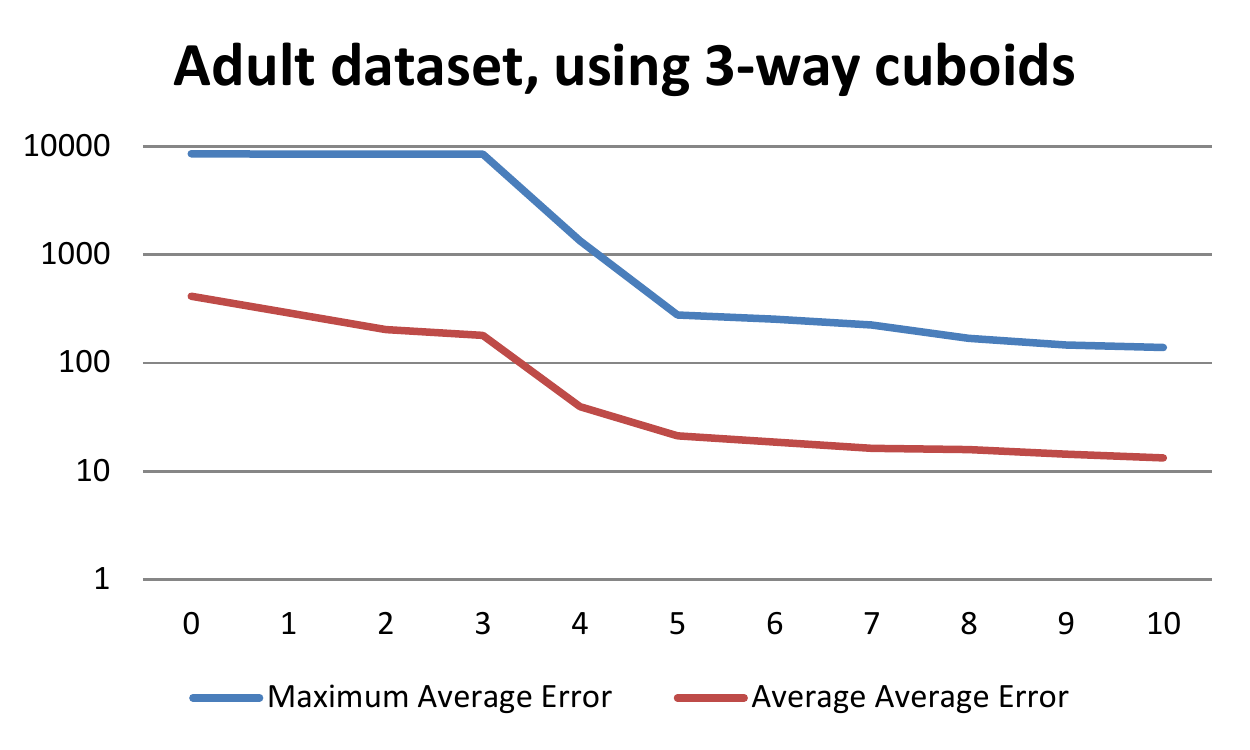}

\caption{Maximum and average error across all 256 cuboids as 10 steps of $\epsilon = 1$ MWEM proceed, where the cuboid error is taken to be the average over its cells of the absolute values of the cell's error.}
\label{fig:cuboids} 
\end{center}
\end{figure}

\begin{figure}[tbp] \begin{center}
\includegraphics[height=4.5cm]{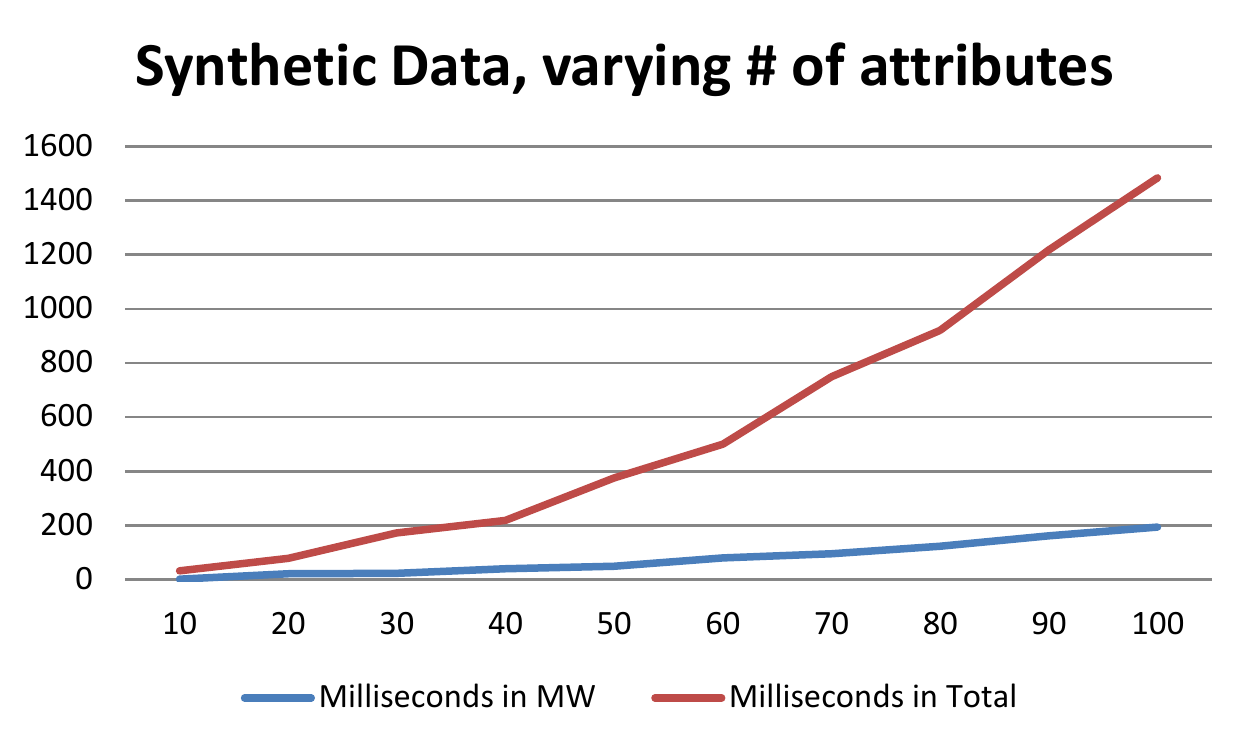}
\includegraphics[height=4.5cm]{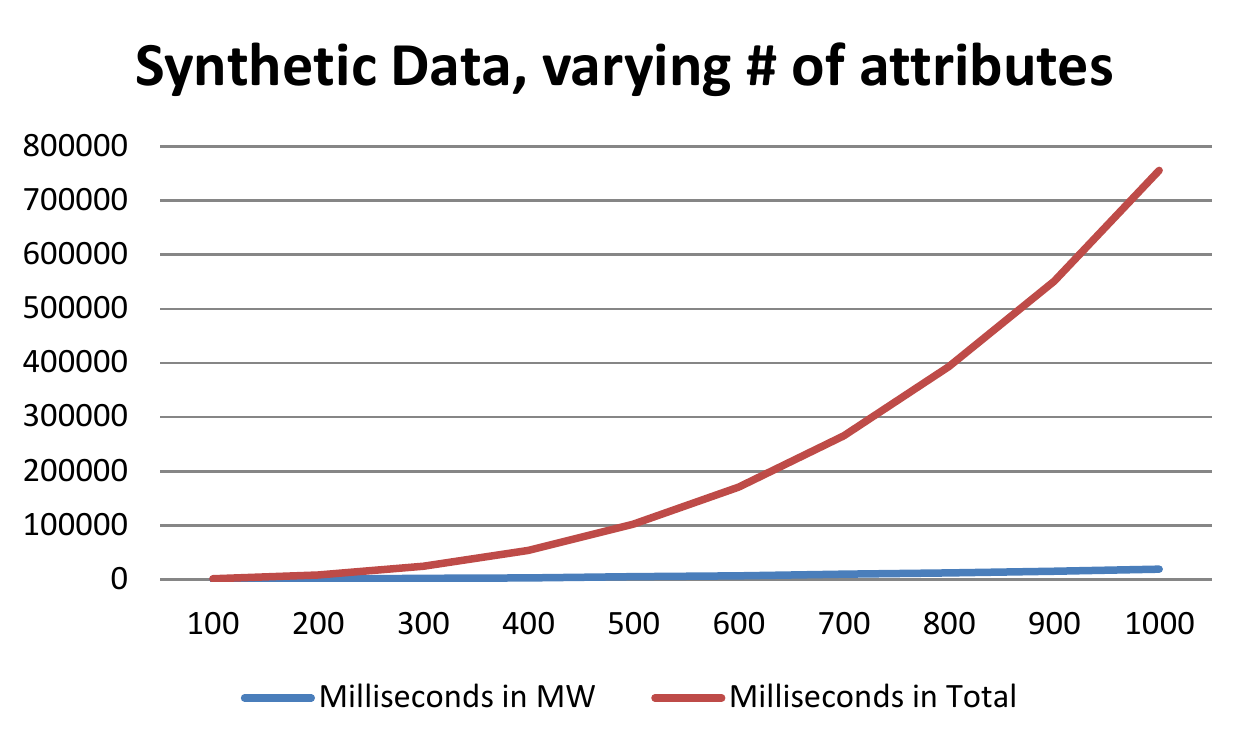}
\caption{Milliseconds spent in MWEM logic, and Milliseconds spent in total, with the latter containing time spent evaluating queries against the source dataset, the $q_i(B)$ evaluations. Even for 1000 binary attributes, we spend only 19 seconds in MWEM.} \label{fig:synthetic}
\end{center} \end{figure}

\begin{figure}[tbp]\begin{center}
\includegraphics[height=4.5cm]{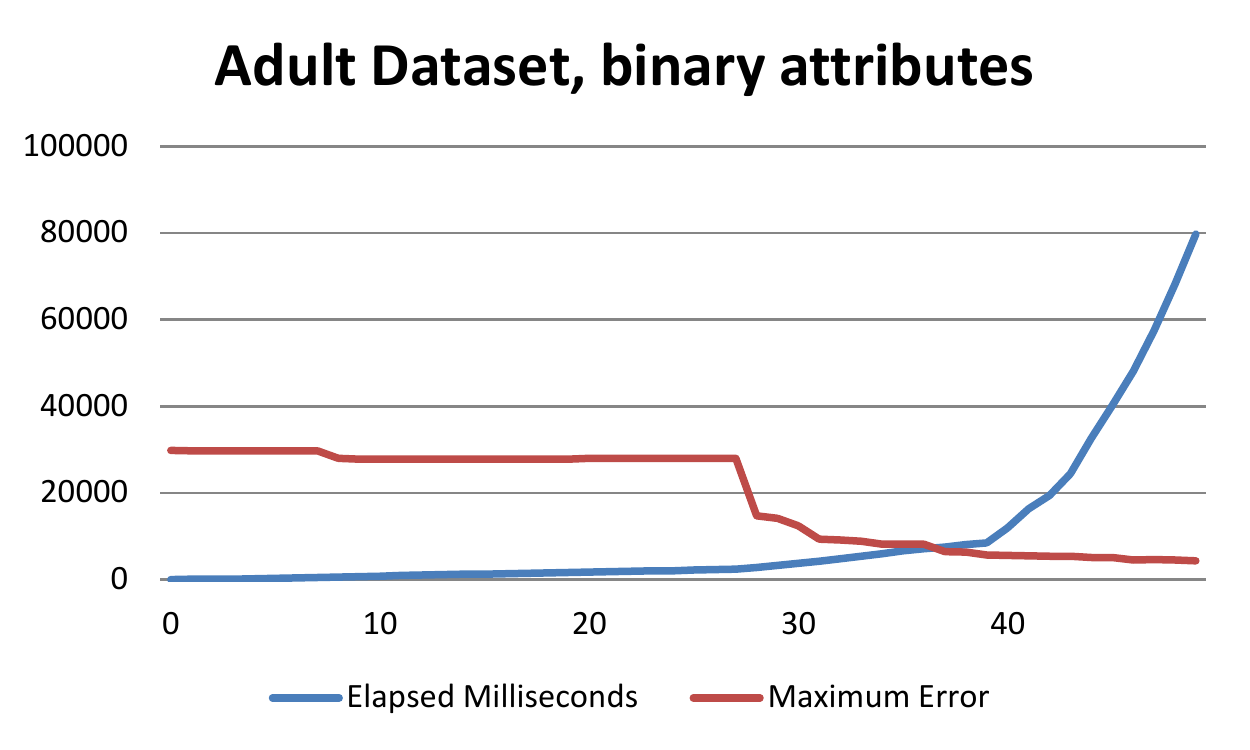}
\includegraphics[height=4.5cm]{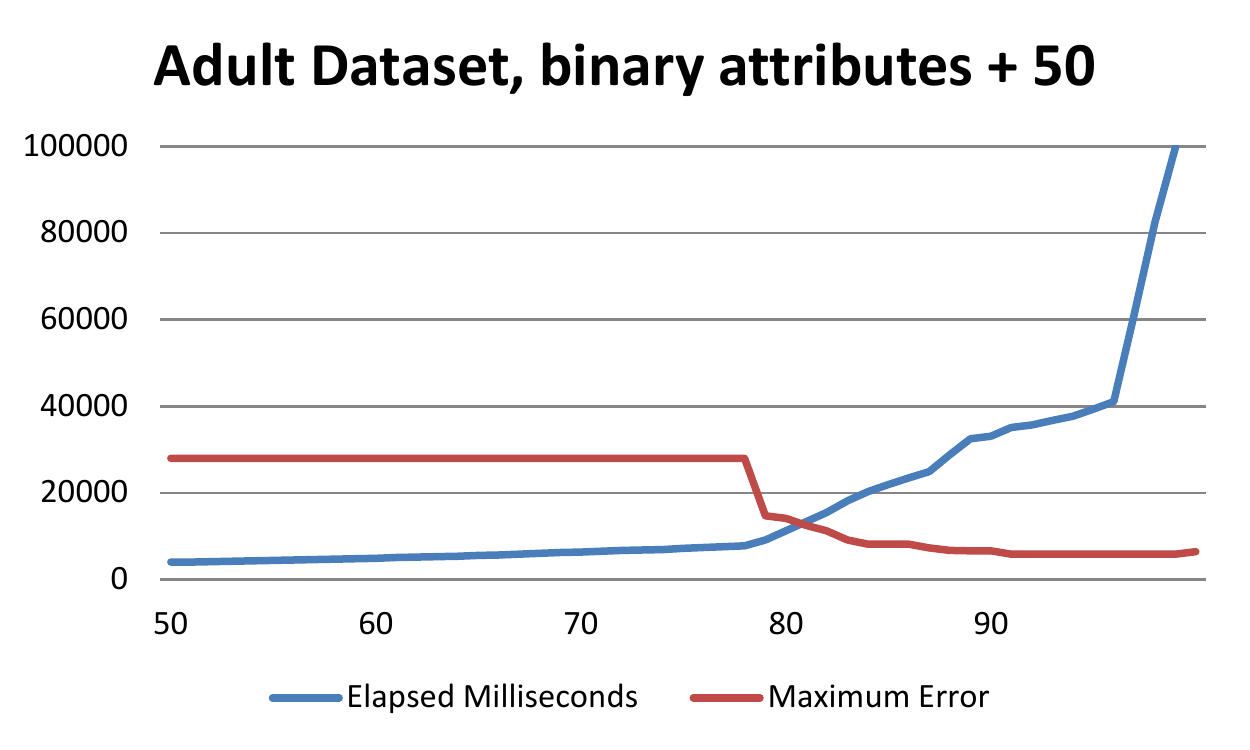}
\caption{ \label{fig:adultbinary}Elapsed milliseconds (increasing) and maximum error (decreasing) as iterations of MWEM proceed. In the second figure, we have added 50 attributes and not plotted the first 50 measurements. The shapes of the curves are very similar, demonstrating how MWEM is capable of ignoring irrelevant attributes. The second experiment takes more time due to the larger number of queries: (77 choose 3) versus (27 choose 3).} 
\end{center} \end{figure}

Fortunately, the distribution we maintain has additional properties we
can exploit. The domain $D$ is often the product of many attributes. If we partition these attributes into disjoint parts $D_1, D_2, \ldots D_k$ so that no query in $H_i$ involves attributes from more than one part, then the distribution produced by Multiplicative Weights is a product distribution over $D_1 \times D_2 \times \ldots D_k$. 
\[
\sum_{x \in D_1 \times D_2 \times \ldots D_k} q(x) \times A_i(x) = \prod_{1 \le j \le k}\left( \sum_{x_j \in D_j} q(x_j) \times A^j_i(x_j)\right) \; .
\]
where $A_i^j$ is a mini Multiplicative Weights over attributes in part $D_j$, using only the relevant queries from $H_i$. 

Importantly, each of the summations above is now over a potentially much
smaller space, reducing the running time from exponential in the number
of attributes in $D$ to the sum of the exponentials in the attributes in
the $D_j$, each of which can be quite small. In the worst case, all
attributes are entangled and we have improved nothing, but for many
realistic datasets independent or irrelevant groups of attributes exist,
and come at essentially no cost in terms of running time or memory
footprint.

\subsection{Scaling Evaluation}

We now experimentally evaluate the scaling performance of this
optimized algorithm on synthetic datasets chosen to highlight its
behavior, and on one real-world dataset in an attempt to sketch how it
might behave in practice. Our experiments are conducted on an AMD
Opteror `Magny-Cours' with 48 processors at 1.9GHz. 

Our first experiment is run on a synthetic dataset containing 100,000
records, over a number of binary attributes varying from 10 to 100, and
then from 100 to 1000. We chose to set each attribute with probability
$p = 0.1$, and $T$ equal to the number of attributes.\footnote{Setting
  the attributes with probability $p = 0.5$, as  in
  \cite{DingWHL11}, results in instantaneous success for MWEM; the
  Exponential Mechanism confirms the uniform distribution as an
  excellent fit and MWEM terminates having done almost no
  work.}  
In Figure \ref{fig:synthetic} we see both the total running time and the
time spent in MWEM. The running time is almost exclusively dominated by
the evaluation of the queries against the private data, $B$, and the
time spent in the factorized implementation of MWEM is essentially
negligible. Each query must be evaluated against $B$, and despite a 48x
speed-up the 100,000 records take more time to process than the factored
MWEM, in which all attributes are in independent components.

Our second experiment evaluates the performance of our factorized
implementation on a binarized form of the Adult dataset where each
attribute is replaced by a number of binary attributes equal to the
logarithm of its range. This results in 27 binary attributes (from 8
discrete attributes). In Figure \ref{fig:adultbinary} we plot the
elapsed running time in milliseconds as a function of the index $i$ of
the MWEM computation (increasing), and the maximum error across all
queries (decreasing). We also repeat the experiment after adding 50 new
binary attributes whose values are set with probability $0.1$, to
demonstrate that our approach successfully ignores irrelevant
attributes, both in terms of running time and maximum error. In this
experiment, MWEM is a significant contributor to the running time. The
absolute running time increases noticeably as the set of measurements
increases in complexity, but until that point each measurement and round
of updates takes less than a second. The actual MWEM complexity stays
below $|D_1| = 2^{27}$ and at no point does it approach $|D_2| =
2^{77}$. 

We can perform the corresponding experiment on the original eight
categorical attributes of Adult, but with so few attributes the $D_j$
are too quickly conflated to give an improvement. We expect the
improvement would be more noticeable on a larger dataset.

\section{Conclusions}
\label{sec:conclusions}
We introduced MWEM, a simple algorithm for releasing data maintaining a high
fidelity to the protected source data, as well as differential privacy with
respect to the records. The approach builds upon the Multiplicative Weights
approach of \cite{HR10,GuptaHRU11}, by introducing the Exponential Mechanism \cite{MT07}
as a more judicious approach to determining which measurements to take. The
theoretical analysis matches previous work in the area, and
experimentally we have evidence that for many interesting settings, MWEM
represents a substantial improvement over existing techniques.

As well as improving on experimental error, the
algorithm is both simple to implement and simple to use. An analyst does not
require a complicated mathematical understanding of the nature of the queries
(as the community has for linear algebra~\cite{LiMi11} and the Hadamard transform~\cite{BCD+07}), but
rather only needs to enumerate those measurements that should be preserved. We
hope that this generality leads to a broader class of high fidelity
differentially-private data releases across a variety of data domains.

\bibliographystyle{plain}
\bibliography{expmult}
\begin{appendices}
\section{Appendix: Proof of Theorem~\ref{UTILITY}}
\label{sec:proof}
The proof of Theorem~\ref{UTILITY} is broken into two parts. First, we argue that across all
$T$ iterations, the queries $q_i$ selected by the Exponential
Mechanism are nearly optimal, and the errors introduced into $m_i$ by
the Laplace Mechanism are small. We then apply the potential function analysis of Hardt
and Rothblum~\cite{HR10} to show that the maximum approximation error for any query cannot be too
large.

Using the shorthand $\err_i \defeq \max_j|q_j(A_{i-1}) - q_j(B)|$ and $\adderr \defeq  2T\log|Q|/\epsilon$, we first claim that with high probability the Exponential Mechanism and Laplace Mechanism give nearly optimal results. 
\begin{lemma}\label{lem:whp}
With probability at least $1-2T/|Q|^{c}$, for any $c \geq 0$, for all $1 \le i \le T$, we have that both
\begin{eqnarray*}
|q_i(A_{i-1}) - q_i(B)| & \ge & \err_i - (2c + 2) \times \adderr \\
|m_i - q_i(B)| & \le & c \times \adderr \; .
\end{eqnarray*}
\end{lemma}

\begin{proof} The probability the Exponential Mechanism with parameter $\epsilon/2T$ selects a query with quality score at least $r$ less than the optimal (meaning the query on which $A_{i-1}$ disagrees the most with $B$) is bounded by
$$ \Pr[|q_i(A_{i-1}) - q_i(B)| < \err_i - r] \le |Q| \times \exp(-\epsilon r/4T). $$ 
If we take $r = (2c + 2) \times 2T\log|Q|/\epsilon$ the probability is at most $1/|Q|^c$ for each iteration.

By the definition of the Laplace distribution,
$$ \Pr[|\textrm{Laplace}(2T/\epsilon)| > r] \le \exp(-r \times \epsilon/2T) \; .$$
Taking $r = c \times 2T\log|Q|/\epsilon$ bounds the probability by at most $1/|Q|^c$ for each iteration.

Taking a union bound over the $2T$ events, we arrive at a failure probability of at most $2T/|Q|^{c}$.
\end{proof}


We next argue that MWEM improves its approximation in each round where $q_i(A) - q_i(B)$ has large magnitude. To capture the improvement, we use the {relative entropy} again:
$$\Psi_i = \sum_{x \in D} B(x) \log (B(x) / A_i(x)) / n \; .$$
The following two properties follow from non-negativity of entropy, and Jensen's Inequality:
\begin{fact}\label{fact:nonneg}
$\Psi_i\ge0$
\end{fact}
\begin{fact}\label{fact:startvalue}
$\Psi_0\le\log|D|$
\end{fact}

We now show that the relative entropy decreases in each round by an amount reflecting the error $q_i$ exposes between $A_{i-1}$ and $B$, less the error in our measurement of $q_i(B)$.
\begin{lemma} 
\label{lem:improvement}
For each round $i \le T$,
\[
\Psi_{i-1}-\Psi_i
\ge \left(\frac{q_i(A_{i-1}) - q_i(B)}{2n}\right)^2 - \left(\frac{m_i - q_i(B)}{2n}\right)^2 \; .
\]
\end{lemma}
\begin{proof}
We start by noting that 
$$\Psi_{i-1}-\Psi_i = \sum_{x \in D} B(x) \log \left( \frac{A_i(x)}{A_{i-1}(x)} \right) / n \; .$$
The ratio $A_i(x)/A_{i-1}(x)$ can be written as $\exp(q_i(x) \eta_i)/\beta_i$, where  $\eta_i = (m_i - q_i(A_{i-1}))/2n$ and $\beta_i$ is the factor required to renormalize in round $i$. Using this notation,
\begin{eqnarray*}
\Psi_{i-1}-\Psi_i 
& = & \eta_i q_i(B)/n - \log \beta_i \; .
\end{eqnarray*}
The required renormalization $\beta_i$ equals
$$ \beta_i = \sum_{x \in D} \exp(q_i(x) \eta_i) A_{i-1}(x) / n  \; . $$
Using $\exp(x) \le 1 + x + x^2$ for $|x| \le 1$, and that $|q_i(x) \eta_i| \le 1$,
$$ \beta_i \le \sum_{x \in D} (1 + q_i(x) \eta_i + q_i(x)^2 \eta_i^2) A_{i-1}(x) / n \; . $$
As $q_i(x)^2 \le 1$, by assumption on all $q_i \in Q$, we have
\begin{eqnarray*}
\beta_i & \le & \sum_{x \in D} (1 + q_i(x) \eta_i + \eta_i^2) A_{i-1}(x) / n  \\
& = & 1 + \eta_i q_i(A_{i-1}) / n + \eta_i^2 \; .
\end{eqnarray*}
Introducing this bound on $\beta_i$ into our equality for $\Psi_{i-1}-\Psi_i$, and using $\log(1 + x) \le x$, we get
$$\Psi_{i-1}-\Psi_i
\ge \eta_i (q_i(B) - q_i(A_{i-1}))/n - \eta_i^2 \; .$$
%
After reintroducing the definition of $\eta_i$ and simplifying, this bound results in the statement of the lemma.
\end{proof}


With these two lemmas we are now prepared to prove Theorem \ref{UTILITY}, bounding the maximum error $|q(A) - q(B)|$.

\begin{proof}[(of Theorem \ref{UTILITY})]
We start by noting that the quantity of interest, the maximum over queries $q$ of the error between $q(A)$ and $q(B)$, can be rewritten and bounded by:
\begin{eqnarray*}
\max_{q \in Q} |q(A) - q(B)| & = & \max_{q \in Q} |q(\avg_{i \le T} A_i) - q(B)| \\
& \le & \max_{q \in Q} \avg_{i \le T} |q(A_i) - q(B)| \\
& \le & \avg_{i \le T} \err_i \; .
\end{eqnarray*}


At this point we invoke Lemma~\ref{lem:whp} with $c = 1$ so that with probability at least $1-2T/|Q|$ we have for $i \le T$ both
\begin{eqnarray*}
\err_i & \le & |q_i(A_{i-1}) - q_i(B)| + 4 \times \adderr \;,\\
|m_i - q_i(B)| & \le & \adderr \; .
\end{eqnarray*}
Combining these bounds with those of Lemma~\ref{lem:improvement} gives
\begin{eqnarray*}
\err_i & \le & \left(4n^2(\Psi_{i-1} - \Psi_i) + \adderr^2\right)^{1/2} + 4 \times \adderr \; .
\end{eqnarray*}
We now average over $i \le T$, and apply Cauchy-Schwarz, specifically that $\avg_i x_i^{1/2} \le (\avg_i x_i)^{1/2}$, giving
\begin{eqnarray*}
\avg_{i \le T} \; \err_i & \le & \left( 4n^2\avg_i(\Psi_{i-1} - \Psi_i)  + \adderr^2\right)^{1/2} \\ 
& & + \; 4 \times \adderr \; .
\end{eqnarray*}
The average $\avg_i (\Psi_{i-1}-\Psi_i)$ telescopes to $(\Psi_0 - \Psi_T)/T$, which 
Facts \ref{fact:nonneg} and \ref{fact:startvalue} bound by $\log(|D|)/T$, giving
\begin{eqnarray*}
\avg_{i \le T} \; \err_i \le \left( 4n^2\log(|D|)/T  + \adderr^2\right)^{1/2} + \; 4 \times \adderr \; .
\end{eqnarray*}
Finally, as $\sqrt{a+b}\le\sqrt{a}+\sqrt{b}$, we derive
\begin{eqnarray*}
\avg_{i \le T} \; \err_i & \le & 2n(\log(|D|)/T)^{1/2} + 5 \times \adderr \; .
\end{eqnarray*}

If we substitute $2T\log|Q|/\epsilon$ for \textrm{adderr} we get the bound in the statement of the theorem. Replacing the factor of $5$ by $(3c + 2)$ generalizes the result to hold with probability $1-2T/|Q|^c$ for arbitrary $c > 0$.
\end{proof}
\end{appendices}

\begin{figure*}
{\smlsize
\begin{verbatim}
double[] MultiplicativeWeightsViaExponentialMechanism(double[] B, Func<int, double>[] Q, int T, double eps)
{
    var n = (double) B.Sum();                                     // should be taken privately, we ignore
    var A = Enumerable.Repeat(n / B.Length, B.Length).ToArray();  // approx dataset, initially uniform
    var measurements = new Dictionary<int, double>();             // records (qi, mi) measurement pairs
    var random = new Random();                                    // RNG used all over the place

    for (int i = 0; i < T; i++)
    {
        // determine a new query to measure, rejecting prior queries
        var qi = random.ExponentialMechanism(B, A, Q, eps / (2 * T));
        while (measurements.ContainsKey(qi))
            qi = random.ExponentialMechanism(B, A, Q, eps / (2 * T));

        // measure the query, and add it to our collection of measurements
        measurements.Add(qi, Q[qi].Evaluate(B) + random.Laplace((2 * T) / eps));

        // improve the approximation using poorly fit measurements
        A.MultiplicativeWeights(Q, measurements);
    }

    return A;
}

int ExponentialMechanism(this Random random, double[] B, double[] A, Func<int, double>[] Q, double eps)
{
    var errors = new double[Q.Length];
    for (int i = 0; i < errors.Length; i++)
        errors[i] = eps * Math.Abs(Q[i].Evaluate(B) - Q[i].Evaluate(A)) / 2.0;

    var maximum = errors.Max();
    for (int i = 0; i < errors.Length; i++)
        errors[i] = Math.Exp(errors[i] - maximum);

    var uniform = errors.Sum() * random.NextDouble();
    for (int i = 0; i < errors.Length; i++)
    {
        uniform -= errors[i];
        if (uniform <= 0.0)
            return i;
    }

    return errors.Length - 1;
}

double Laplace(this Random random, double sigma)
{
    return sigma * Math.Log(random.NextDouble()) * (random.Next(2) == 0 ? -1 : +1);
}

void MultiplicativeWeights(this double[] A, Func<int, double>[] Q, Dictionary<int, double> measurements)
{
    var total = A.Sum();

    for (int iteration = 0; iteration < 100; iteration++)
    {
        foreach (var qi in measurements.Keys)
        {
            var error = measurements[qi] - Q[qi].Evaluate(A);
            for (int i = 0; i < A.Length; i++)
                A[i] *= Math.Exp(Q[qi](i) * error / (2.0 * total));

            var count = A.Sum();
            for (int i = 0; i < A.Length; i++)
                A[i] *= total / count;
        }
    }
}

double Evaluate(this Func<int, double> query, double[] collection)
{
    return Enumerable.Range(0, collection.Length).Sum(i => query(i) * collection[i]);
}
\end{verbatim}
}
\caption{Full source code for a reference implementation of MWEM.
\label{fig:code}
}\end{figure*}

\end{document}